%% file: dynamic.tex
\documentclass[letterpaper,11pt]{article}

\usepackage{wrapfig}

\usepackage{amsmath,amsfonts,amssymb,amstext,amsthm}

\usepackage{url}
\usepackage{float}
\usepackage[margin=1in]{geometry}
\usepackage[usenames]{color}
\usepackage[dvipsnames,usenames,table]{xcolor}
\usepackage[colorlinks=true,urlcolor=blue,linkcolor=RoyalBlue,citecolor=OliveGreen]{hyperref}
\usepackage{cleveref}

\usepackage{enumitem}

\usepackage{verbatim}
\usepackage{xspace}

\usepackage[normalem]{ulem}

\usepackage{tikz}
\usetikzlibrary{decorations.markings}
\usepackage{graphicx}
\usepackage{multirow}
\newtheorem{remark}{Remark}
\newtheorem{thm}{Theorem}
\newtheorem{cor}[thm]{Corollary}
\newtheorem{lem}[thm]{Lemma}
\newtheorem{fact}[thm]{Fact}

\newtheorem{claim}[thm]{Claim}
\newtheorem{obs}[thm]{Observation}
\newtheorem{definition}[thm]{Definition}

\newtheorem{example}[thm]{Example}

\newcommand{\expectation}[1]{\mathbb{E} \left[ #1 \right]}

\def\squarebox#1{\hbox to #1{\hfill\vbox to #1{\vfill}}}

\usepackage{sidecap,caption}
\usepackage{complexity}

\usepackage{algorithmic}
\usepackage[ruled]{algorithm2e}
\usepackage{hyperref}
\usepackage{cleveref}
\usepackage{thmtools}
\usepackage{thm-restate}

\begin{document}

				
\title{On the Complexity of Dynamic Mechanism Design}

\author{Christos Papadimitriou\thanks{Columbia University; christos@cs.columbia.edu} 
\and George Pierrakos\thanks{Hudson River Trading; geopier@gmail.com}
\and Alexandros Psomas\thanks{Purdue University; apsomas@cs.purdue.edu}
\and Aviad Rubinstein\thanks{Stanford University; aviad@cs.stanford.edu}}

\maketitle
		
\begin{abstract}
 
We introduce a simple dynamic mechanism design problem in which the designer offers two items in two consecutive stages to a single buyer. The buyer's joint distribution of valuations for the two items is known, and the buyer knows the valuation for the current item, but not for the one in the future.  The designer seeks to maximize expected revenue, and the mechanism must be deterministic, truthful, and ex-post individually rational. We show that finding the optimum deterministic mechanism in this situation --- arguably one of the simplest meaningful dynamic mechanism design problems imaginable --- is \NP-hard.  We also prove several positive results, including a polynomial-time linear programming-based algorithm for the revenue optimal randomized mechanism (even for many buyers and many stages). We prove strong separations in revenue between non-adaptive, adaptive, and randomized mechanisms, even when the valuations in the two stages are independent. Finally, for the same problem in an environment in which contracts cannot be enforced, and thus perfection of equilibrium is necessary, we show that the optimum randomized mechanism requires multiple rounds of cheap talk-like interactions.
\end{abstract}

\input{intro}

\section{The mechanism}
\label{sec:prelim}
\input{prelim}


\section{Deterministic mechanisms are \NP-hard}
\label{sec:reduction}
\input{reduction}

\section{Deterministic cases solvable in polynomial time}
\label{sec:special_deterministic}
\input{deterministic}

\section{Randomized adaptive mechanisms: Multiple stages, multiple buyers}
\label{sec:LP}
\input{LP3}

\section{Separations}
\label{app:separations}
\input{different_auctions}

\input{no_contract}

\section{Discussion and Future Work}
\input{discussion}

\section*{Acknowledgments}

We thank the GEB guest editor,  Brendan Lucier, and anonymous referees for many helpful comments.
Christos Papadimitriou's research was partially supported by NSF awards CCF1763970 and CCF1910700, and by a research contract with Softbank.
Alexandros Psomas is supported by a Google Research Scholar Award.
Aviad Rubinstein is supported by NSF CCF-1954927 and CCF-2112824, and a David and Lucile Packard Fellowship.  

\bibliography{refs}

\input{appendix}

\end{document}

%% file: intro.tex

\section{Introduction}

Consider the problem of a revenue maximizing seller with two items for sale, one today and one tomorrow, to a single buyer. The buyer knows her value for today's item, but for tomorrow's item she only has a prior. The seller knows the joint distribution (for which the buyer's prior is the conditional).  How should the seller behave in order to maximize expected revenue?
If there was no item tomorrow, this would be a simple application of Myerson's theorem \cite{myerson1981optimal}:  the seller makes an offer easily calculated from the buyer's prior.  But, the second item makes things much more complicated.  We have a {\em dynamic} mechanism design problem.


Dynamic mechanisms have been studied extensively in quite general settings; see for example \cite{bergemann2011dynamic} and \cite{bergemann2019dynamic} for recent surveys. But, in contrast to previous work, our focus here is \emph{computation}. We propose the two-stage mechanism problem as a useful surrogate of dynamic mechanisms for the purpose of exploring the problem's computational complexity.  It is certainly extremely simple, and yet surprisingly hard.  To see why, suppose that the two valuations, for today and tomorrow, are independent random variables.  It is then tempting to assume that, in this simple case, running Myerson's mechanism in each round should work (we refer to this as the \emph{non-adaptive} mechanism). Is it optimal among all truthful and ex-post individually rational mechanisms? Even if not, it must surely at least be a good approximation? The answer is ``no''!   

\begin{example}\label{example 1}
Let $X_1$ and $X_2$ be the random variables indicating the value of the buyer for the first and second stage item. $X_1$ takes value $2^i$ with probability $2^{-i}$ for $i= 1,\dots,n$, and value $0$ with probability $2^{-n}$. $X_2$ takes value $2^i$ with probability $2^{-i}$ for $i=1,\dots,2^n$, and value $0$ with probability $2^{-2^n}$. It can be verified that the optimal static mechanism for both $X_1$ and $X_2$ extracts revenue at most $2$: Consider setting some price $2^k$. The expected revenue is at most $2^k \cdot \sum_{i \geq k} 2^{-i} \leq 2$. Therefore, running the optimal static mechanism at each stage extracts revenue at most $4$. On the other hand, there exists a dynamic mechanism that extracts revenue $n$: on the first stage the buyer pays her report $v^{(1)}$. On the second stage the item is given for free with probability $\frac{v^{(1)}}{\expectation{X_2}}$. Notice that $\expectation{X_2} > 2^n$, therefore $\frac{v^{(1)}}{\expectation{X_2}}$ is a probability. An easy calculation shows that truthful reporting (weakly) maximizes the buyer's utility. The revenue extracted is $\expectation{X_1} = n$. The intuition here is that if the expected value of the future item is large, the buyer is willing to pay her value on the first stage for a better probability of getting allocated the future item.
\end{example}


We note that similar behaviors have been exhibited before, e.g. see~\cite{courty2000sequential, krahmer2016optimality}. And, in fact, for the exact same valuations as Example~\ref{example 1}, we can extract more revenue than the non-adaptive mechanism using a deterministic mechanism.  Here, we overall prove that, even when restricted to ex-post IR mechanisms, there is revenue loss by a non-constant factor between: non-adaptive mechanisms and the optimum deterministic adaptive mechanism (even for uncorrelated distributions);  the optimum deterministic and the optimum randomized mechanism; the optimum randomized mechanism and the optimum social welfare. See Section~\ref{app:separations} for the precise statements and proofs. 

\subsection*{Our results}

Let us focus on deterministic, ex post individually rational mechanisms. How hard can it be to find the optimum one?  The reason we insist on determinism and ex post IR is because we believe that they draw the boundary of mechanisms in which people are likely to choose to participate. That is, the ex post IR constraint rules out mechanisms in which the seller asks an advance payment equal to the expected surplus. Ruling out such mechanisms (which would be feasible under ex-ante or ex-interim IR constraints) is of practical importance. For example, in online ad auctions advertisers' bids are interpreted as willingness to pay, and typically the platform cannot (legally) charge advertisers more than the amount declared as maximum. Our main result (Theorem~\ref{thm:2-day-np-hard}) is that, for the two stage case, it is strongly  {\NP-complete}, given a prior with finite support, to find the optimum such mechanism. However, as we discuss below, we do show that we can compute the optimal randomized mechanism in polynomial time.\footnote{A similar contrast was noted in the case of Myerson's mechanism with correlated buyers, see~\cite{dobzinski2011optimal},\cite{papadimitriou2011optimal}.}

We first characterize the mechanisms of interest (Section \ref{sec:prelim}).  It turns out that there is always an optimum deterministic mechanism that is {\em semi-adaptive} (Lemma~\ref{det-semi-adaptive}).  A non-adaptive mechanism is one that makes two independent offers, one now and one in the future, without eliciting any input from the buyer.  In contrast, a semi-adaptive mechanism starts by eliciting the buyer's type (and takes care that she is truthful), and then makes two offers simultaneously, one for now and one for the future.  The buyer can take or leave the first offer now, and come back in the second stage to take or leave the second offer (which she knows now). So, our task is reduced to designing a function, informed by the whole joint distribution, that maps the support of the buyer's first stage type distribution to two prices. 
One of the reasons this task is daunting is that truthfulness is quite subtle in this context, and incentive compatibility constraints are a big part of the problem's difficulty.  We must give the buyer the right incentives (both right now and in future expectation) so she will not misrepresent her type.  This is done by choosing price pairs such that, for any other current valuation, the buyer is best off, in expectation, telling the truth.  Low prices now must be counterbalanced carefully with higher prices in the future, and the inequalities involve integrals of the cumulative conditional distributions of the future valuation.

In Section \ref{sec:reduction} we describe our \NP-completeness proof, from {\sc Independent Set}.  It is quite elaborate.  The first stage types (values in the support of the first stage distribution) are the nodes.  For each type, two of the possible current prices stand out as potentially optimal, and choosing between them is tantamount to deciding whether a node will be in the maximum independent set.  The optimum revenue achieved is a strictly increasing function of the independent set size.  The truthfulness constraints enforce that no two adjacent nodes are included, and this necessitates an elaborate design of the conditional distributions associated with the nodes.

Before we proceed to our next set of results, it is worth pointing out that the source of the complexity of the two-stage mechanism is not the multi-dimensionality of the buyer's private information. Finding the optimal deterministic one-shot mechanism for selling two (possibly correlated) items to a single buyer is computationally tractable (\cite{chen2018complexity}).\footnote{And in fact, it remains tractable for any constant number of items.} The key fact that allows for good algorithms is that for a constant number $k$ of items, the optimal mechanism for any number $N$ of types (where a type here is a vector of $k$ values) has at most $d=2^k$ possible prices, one for each bundle of items, which is again a constant. This allows to partition the $d$-dimensional space of price vectors into cells, such that for each cell the buyer has the same behavior. In contrast, in the dynamic problem the optimal deterministic mechanism for $2$ items and $N$ types (where a type is a vector of two values) might have $N$ different price pairs (one price for each item), which is not a constant. In fact, in our construction we have two candidate price pairs for each type, leading to an exponential number of candidate solutions. Thus, the dynamic problem is computationally much harder than the static one. Furthermore, as we prove later, the dynamic problem is tractable (at least for two items) when stages are independent, as well as when randomization is allowed. Therefore, correlation and determinism are both necessary for our hardness result. Finally, even though we do not know whether computing the optimal deterministic ex-ante IR mechanism is tractable for correlated stages, notice that this task is trivial for independent stages, even for $D>2$ items: the optimal mechanism offers a take-it-or-leave-it price for the first item (recall that the buyer knows her value for this one) and all the remaining items for a take-it-or-leave-it price equal to their expected value (which will be always accepted by a risk-neutral buyer).

In Section \ref{sec:special_deterministic} we present positive results for deterministic mechanisms.  First, if the support of the distribution of the first stage valuation is a constant, then the problem becomes easy: once we have fixed the second stage prices (there is only a constant number of them), we can easily optimize the first stage prices, by writing a simple linear program. A similar simple linear program suffices even when the second stage prices are not fixed, but we have decided between which second stage types each price should be in. The overall number of linear programs we need to solve is the number of second stage types raised to the power of first stage types, which is polynomial. Second, if we are given the first stage prices and we want to optimize the second stage prices, we can find a $(1-\epsilon)$ approximately optimal mechanism in time polynomial in $1/\epsilon$ and the size of the input, for all $\epsilon > 0$ (that is, there exists an FPTAS), based on an {\em integer} program that happens to be totally unimodular.  These two positive results point to the source of one major difficulty in proving \NP-completeness of the problem:  in our construction the prices of {\em both} items must vary over types. Our last positive result for deterministic mechanisms is a polynomial time algorithm for computing the optimal deterministic mechanism when the stages are independent. The two driving factors of this result are (1) the allocation of the first stage item is a monotone function, and, (2) the first stage price of a type $t_i$ is either zero, or her valuation for the first item. Since all first stage types have the same second stage distribution (since the stages are independent), the IC constraints severely limit the amount that we can price discriminate between different first stage types. Combining this observation with the two aforementioned facts we can show that the number of optimal mechanisms is a small polynomial: a simple enumeration is computationally tractable.

We proceed to study randomized mechanisms. A deterministic and optimal mechanism can often be expressed as the solution to an integer program (but finding an optimal solution to an integer program is  an \NP-complete problem). The relaxed program,\footnote{The relaxation of an integer program arises by replacing integrality constraints, e.g. a constraint of the form $x_i \in \{ 0, 1 \}$, with linear constraints, e.g. $x_i \in [0,1]$.} which can be solved in polynomial time, typically encodes the optimal randomized mechanism.
In Section \ref{sec:LP} we show that the problem of finding the optimum randomized mechanism can be solved in time polynomial in the number of types (where a type here specifies the buyer's values for all items, across stages), and in fact for any finite number of stages of sale and for any constant number of buyers. We show how to optimize over mechanisms that satisfy a number of different notions of incentive compatibility (some of which are with loss of generality, but, as we argue, might be worth optimizing over). Several reasons why this LP should be have an exponential number of variables and constraints must be overcome.  For example, our mechanisms elicit from each buyer and each stage only the value of the buyer for the item in that stage. Thus, for our choice of incentive compatibility, naively there are exponentially many ways a buyer can misreport: when strategizing about what to report on the current stage, our buyers consider all (exponentially many) functions from future realizations of her value to future reports. Carefully defining the IC constraints by backward induction resolves this issue.  A second issue is that one seems to need to ``remember'' in each stage the utility accumulated so far for each buyer in order to achieve ex-post individual rationality, which requires that the total utility of a buyer is non-negative at the end, once all uncertainty has been resolved. We overcome this obstacle by reducing ex-post IR mechanisms to stage-wise ex-post IR mechanisms (the utility in each stage is non-negative).

All these results imply that the seller can increase her revenue by committing to a specific future behavior, presumably through a contract. In Section \ref{sec:No-contract} we consider a closely related question, {\em what is the revenue-optimal design when contracts about future behavior cannot be written and enforced,} and thus the seller cannot commit to an irrational behavior in the future, such as doing something that may be suboptimal at the time? 
We model this question as a form of (rather benign%
\footnote{In the conference version of this paper we called this setting {\em no commitment}. Here we change the name to be consistent with the rich Economics literature on limited commitment.}) {\em limited commitment} where in the second stage the seller will observe all previous communication, update her prior accordingly, and then run the current optimal mechanism (Myerson's mechanism on the updated prior).
For the first stage, the seller designs a truthful%
\footnote{As one should suspect, asking the buyer to report her type on the first stage is with loss of generality; by {\em truthful} we mean that it is incentive compatible for the buyer to follow the protocol.} mechanism and communication protocol.
In this setting, we demonstrate a different facet of the complexity of dynamic mechanisms:  The revenue-optimal randomized mechanism requires the seller and buyer to interact through {\em multiple rounds of communication} in the first stage (we can prove three, and we conjecture an unbounded number).  The idea is that the seller can offer during the first stage, along with the first item, a second ``product'', the {\em OTR}: an opportunity for the buyer to truthfully report more information about her valuation on the second stage. This leads to multiple rounds of communication: bidding for the first item, allocation of the OTR with some probability, reporting the second stage information (if given the OTR), and so on. We see this phenomenon as a different aspect of the difficulty of dynamic mechanism design.

See Table~\ref{table:summary} for a summary of our results.

\begin{table}[ht]
\centering
    \begin{tabular}{ | p{5cm} | p{10cm} |}
    \hline
    Setting & Result \\ \hline
    $2$ correlated stages, $1$ buyer &  Thm~\ref{thm:2-day-np-hard}: Finding the optimum deterministic mechanism is NP-hard. \\ \hline
    $2$ correlated stages, $1$ buyer&  Thm~\ref{thm:FPTAS-unimodular}: If the prices in the first stage are fixed, then the optimum deterministic mechanism can be approximated by an FPTAS. \\ \hline
    $2$ correlated stages, $1$ buyer &  Thm~\ref{thm:constant-support}: If the support of first stage valuations is constant, then the optimum deterministic mechanism can be computed in polynomial time. \\ \hline
    $2$ independent stages, $1$ buyer &  Thm~\ref{thm:independent}: The optimum deterministic mechanism can be computed in polynomial time. \\ \hline
    $D$ correlated stages, $k$ buyers (constant) &  Thm~\ref{thm:lp-manydays-manybidders}: The optimum randomized mechanism can be computed in polynomial time in the number of types and in the number of stages. \\ \hline
    $2$ correlated stages, $1$ buyer &  Thm~\ref{thm:separations}: Separations between social welfare and optimal non-adaptive, adaptive deterministic and adaptive randomized mechanisms. \\ \hline
    $2$ independent stages, $1$ buyer, no contracts &  Thm~\ref{thm:no-contract}: The optimum randomized mechanism requires multiple rounds of communication. \\ \hline    
\end{tabular}
\caption{Summary of main results}
\label{table:summary}
\end{table}

\input{related}

%% file: related.tex
\subsection{Related and Subsequent work}

We briefly discuss research in dynamic mechanism design that is most related to the current work. For an extensive review of the literature see~\cite{bergemann2011dynamic} and \cite{bergemann2019dynamic}. The study of revenue maximization in an environment where the agent's private
information changes over time was initiated by~\cite{baron1984regulation}. Revenue optimal dynamic auctions are studied more recently by~\cite{courty2000sequential}, \cite{esHo2007optimal} and \cite{pavan2009dynamic, pavan2010infinite, pavan2014dynamic}. A common constraint in these works is that the principal has to satisfy all of the sequential incentive constraints, but only a single ex-ante participation constraint.

\paragraph{Limited commitment}
There is a long literature on mechanism design with limited commitment, starting from~\cite{laffont1987comparative, laffont1988dynamics}.
Two recent works (published after the conference version of our paper) that are particularly closely related to the no-contract model we study in Section~\ref{sec:No-contract} are the limited commitment models of~\cite{lobel2019dynamic} and~\cite{doval2018mechanism}. \cite{lobel2019dynamic} study a positive commitment model where the seller can make promises regarding future allocations, but not future mechanisms; that is, the seller can commit to allocating a future item to the buyer, but can make no promises regarding the future mechanism. In our work the seller can instead make promises regarding the way she will update her prior, but no promises about future allocations and payments can be made.~\cite{doval2018mechanism} allow the seller to design a \emph{communication device} that the buyer can interact with. The buyer's input to the communication device is not observed by the seller (and thus, if she is incentivized to do so the buyer can report her full type to the device), and instead the communication device sends a message $m$ to the seller. This message is the only piece of information that the seller has going into the second stage. In our work we study direct communication between the seller and buyer, that is, everything that the buyer signals is accessible to the seller. Thus, as opposed to the setting of~\cite{doval2018mechanism} multiple rounds of communication in the first stage can (and do) benefit the seller. We note that the communication device of~\cite{doval2018mechanism} can simulate infinite rounds of communication in the setting studied here.

Limited commitment mechanisms have also been studied in the context of sellers with private information, e.g.~\cite{Hrner2010SellingI, BabaioffKL12}. Specifically,~\cite{BabaioffKL12} consider the seller's optimal mechanism for selling information; this mechanism depends on the information that the buyer wants to acquire, so it can't immediately be revealed to the buyer, which leads to multi-rounds mechanisms. In contrast, in our model the seller does not have any private information and the additional revenue from interactive mechanisms arises because the buyer {\em wants to report} more information about their private type to incentivize seller's behavior in the second stage.

\paragraph{Subsequent work in Computer Science} Subsequent to the preliminary version of this paper, \cite{ashlagi2016} provide characterizations of the optimal ex-post IR, periodic incentive compatible dynamic mechanism, with $m$ \textit{independent} stages and $n$ buyers. \cite{ashlagi2016} show that there exists an optimal mechanism that has stage utility equal to zero for all stages, except maybe the last, where the seller might have to pay the buyers. Surprisingly, their mechanism can be described via updates, at every stage, to a scalar variable that guides the future allocation and payments. The authors use this characterization to give a mechanism that obtains a $\frac{1}{2}$ approximation to the optimal revenue for the single buyer problem.

\cite{mirrokni2016dynamic} study dynamic mechanisms with an interim IR constraint. They define a class of mechanisms called \textit{bank account} mechanisms. Bank account mechanisms maintain a state variable, the balance, that is updated throughout the execution of the mechanism depending on a ``spending'' and ``depositing'' policy. The allocation and payment at each stage depend on the report and the balance.~\cite{mirrokni2016optimal} study revenue maximization for bank account mechanisms subject to an ex-post IR constraint. In Section~\ref{sec:LP} we use their reduction from ex-post IR mechanisms to stage-wise ex-post IR mechanisms in order to simplify our linear program.

\cite{mirrokni2020non} study the design of \textit{non-clairvoyant} dynamic mechanisms. An oblivious dynamic mechanism decides on the allocation and payment for stage $k$ using information only about the current and past stages, i.e. it is oblivious about the buyers' value distributions $D_{k+1},\dots,D_m$. Their mechanism \textit{ObliviousBalance} runs at each stage a combination of Myerson's optimal auction, a second price auction, and the money burning mechanism of~\cite{hartline2008optimal}. Their mechanism obtains a $\frac{1}{5}$ approximation to the optimal revenue. 

\cite{mirrokni2019optimal} show that optimal dynamic auctions are virtual welfare maximizers, under some definition of virtual welfare. Specifically, in each stage $d$ the optimal dynamic auction is a second price auction on an appropriately defined virtual value space. In order for the virtual welfare maximizing allocation rule to be monotone, ironing is necessary, but unlike ironing in Myerson's optimal auction, the ironing
step is interdependent across the values of different buyers.

\cite{liu2018competition} study prior-independent dynamic mechanisms, and more specifically, they show an analogue of the classic Bulow-Klemperer result in auction theory. $m$ items are auctioned off in $m$ consecutive stages to $n$ independent and identical buyers.
They show that recruiting $3n$ more buyers and executing a simple second price auction at each stage yields more revenue than the optimal dynamic auction, even when the buyers’ values are correlated across stages, under a monotone hazard rate assumption on the stage (marginal) distributions. This result can be turned into a $4$-approximation algorithm by simulating the $3n$ additional buyers. For the general case, beyond marginals that have monotone hazard rate, we are not aware of any algorithms that give a constant approximation,  even for a single buyer and two correlated stages.

\cite{agrawal2018robust} study revenue maximization for a buyer who is not fully rational, but instead uses some specific form of learning behavior. They give a simple state-based mechanism that gives simultaneously a constant approximation to revenue extracted by the optimal auction for a $k$-lookahead buyer for all $k$, a buyer who is a no-regret learner, and a buyer who is a policy-regret learner.

%% file: prelim.tex

The {\sc Two-stage Mechanism} problem involves  a seller with $2$ items to sell to a single buyer. The items are sold in two consecutive stages, one item per stage. The buyer privately learns her types over time. In the beginning of stage $i$ she learns her type for that stage. The buyer can have one of $|V^{(1)}|$ types in the first stage.  The $i$-th first stage type $t_i$ occurs with probability $Pr[t_i]$. A buyer with first stage type $t_i$  has valuation $v^{(1)}_i$ for the first item, and a probability distribution $f_i$ over valuations/second stage types for the second item. We assume that each type $t_i$ has a different first stage valuation $v^{(1)}_i$, and therefore we use type and valuation interchangeably.\footnote{We note that this restriction makes the problem computationally easier.} The joint distribution is known to the seller and the buyer. We write $v^{(2)}$ for the valuation of the second stage item, and $V^{(2)}$ for the support of the distribution in stage $2$. We assume that $0$ is always in the support of $f_i$, for all $i$.

The order of events is as follows: (1) The buyer privately learns her type $t_i \in V^{(1)}$ for stage $1$, and sends a message to the mechanism, (2) the seller implements an allocation $x_1 \in \{ 0, 1 \}$ for item $1$ and charges a payment $p$, (3) the buyer obtains stage utility $u_1 = v^{(1)}_i \cdot x_1 - p$, (4) the buyer privately learns her value (second stage type) $v^{(2)}$ for the second stage item, and sends a message to the mechanism, (5) the seller implements an allocation $x_2 \in \{ 0, 1 \}$ for item $2$ and charges a payment $q$, (6) the buyer obtains stage utility $u_2 = v^{(2)} \cdot x_2 - q$. The buyer's overall utility is $u_1 + u_2$, i.e. the buyer is additive without discount.  
Motivated by practical considerations (e.g. in advertising auctions bidders must submit a bid, not an abstract signal/message) we restrict the set of messages the buyer can send to the mechanism to be types, i.e. the buyer sends a type $b^{(1)} \in V^{(1)}$ and a value $b^{(2)} \in V^{(2)}$ in the first and second stage, respectively.

We focus on dynamic, direct and deterministic mechanisms.  A deterministic mechanism for this problems consists of an allocation, payment rule pair $(x_1, p)$ for stage $1$, and an allocation, payment rule pair $(x_2,q)$ for stage $2$. $x_1$ and $p$ map $V^{(1)}$ to $\{ 0 , 1 \}$ and $\mathbb{R}$, respectively. $x_2$ and $q$ map $V^{(1)} \times V^{(2)}$ to $\{ 0 , 1 \}$ and $\mathbb{R}$, respectively. That is, the allocation and payment in the second stage can depend on the report in the first stage. Note, however, that the buyer does not report all the values she has observed so far in each stage $2$, but only her value in stage $2$. Our goal is to design a mechanism that maximizes the seller's expected revenue, subject to \textit{dynamic incentive compatibility} and \textit{ex post individual rationality}.

The dynamic revelation principle~\cite{myerson1986multistage,sugaya2021revelation} states that there is no loss of generality in restricting attention to dynamic direct mechanisms where buyers report their information truthfully (the buyer's reported type coincides with her true type). However, the dynamic revelation principle does not require that a buyer reports their information truthfully in the second stage after a lie in the first stage (see~\cite{pavan2014dynamic} for an application of the dynamic revelation principle in a similar context). As we see later in this section, our mechanisms will satisfy a stronger notion of incentive compatibility, where a buyer reports her true value in the second stage, no matter what the report in the first stage was. This should be intuitively obvious since the second stage utility $v^{(2)} x_2( t_i, v^{(2)} ) - q( t_i, v^{(2)} )$ is unaffected by the first stage type $t_i$ (beyond its effects on the mechanism itself).
Dynamic incentive compatibility (DIC) can be defined by backward induction. In the last stage, assuming honest reports so far, it should be incentive compatible for the buyer to report her true type (using the standard notion of incentive compatibility in static mechanism design). That is, assuming an honest first stage report $t_i \in V^{(1)}$, and all $v^{(2)} \in V^{(2)}$ (such that $v^{(2)}$ is drawn from $f_i$ with positive probability), and for all $b^{(2)} \in V^{(2)}$ we have
\begin{equation}\label{eq: second stage ic}
v^{(2)} x_2( t_i, v^{(2)} ) - q( t_i, v^{(2)} ) \geq v^{(2)} x_2( t_i, b^{(2)} ) - q( t_i, b^{(2)} ).
\end{equation}
Then, in the first stage, it should be incentive compatible for the buyer to report her true type. Since the second stage mechanism satisfies Equation~\ref{eq: second stage ic}, when calculating her expected future utility after being honest in stage one, she should assume that she will report her true second stage type honestly as well.
Let $u_2( t' ; v^{(2)}, b^{(2)} ) = v^{(2)} x_2( t', b^{(2)} ) - q( t', b^{(2)} )$ be the second stage utility when the buyer with second stage value $v^{(2)}$ reports $b^{(2)}$, and her first stage report was $t'$. For all $t_i, t' \in V^{(1)}$ we have that
\[
v^{(1)}_i x_1( t_i ) - p( t_i ) + \mathbb{E}_{v^{(2)} \sim f_i} \left[ u_2( t_i; v^{(2)}, v^{(2)} ) \right] \geq v^{(1)}_i x_1( t' ) - p( t' ) + \mathbb{E}_{v^{(2)} \sim f_i} \left[ \max_{b^{(2)}} u_2( t'; v^{(2)}, b^{(2)} ) \right].
\]

A mechanism is ex-post individual rational (ex-post IR)  if it guarantees non-negative utility for the buyer. That is, the buyer's utility should be non-negative in all outcomes output by the mechanism, assuming she reports truthfully. For stage $2$ this implies that  for all $t_i\in V^{(1)}$ and all $v^{(2)} \in V^{(2)}$ the mechanism satisfies 
\[
v^{(1)}_i x_1( t_i ) - p( t_i ) + v^{(2)} x_2( t_i, v^{(2)} ) - q( t_i, v^{(2)} ) \geq 0. 
\]
Since $v^{(2)}$ occurs with non-zero probability for all $t_i$, it must be that in the first stage allocation and payments satisfy $v^{(1)}_i x_1( t_i ) - p( t_i ) - q( t_i, 0 ) \geq 0$. Without loss of generality we therefore have that $v^{(1)}_i x_1( t_i ) - p( t_i ) \geq 0$ (since, if $q( t_i, 0 )$ is negative, we can always increase it to zero, and decrease $p( t_i )$ appropriately, without affecting incentives).
As we will see in Section~\ref{sec:LP}, and as already shown by~\cite{mirrokni2016optimal}, an ex-post IR mechanism can be turned into a stage-wise ex-post IR mechanism (non-negative utility in each stage) with at least as much expected revenue.

\subsection{Semi-adaptive mechanisms}

What can we say about the structure of deterministic revenue-optimal  dynamic mechanisms? The point of this paper is that they are quite complex. Nonetheless, we can significantly restrict our search space. 
So far we have allowed mechanisms to be {\em adaptive}: the allocation and payment in the second stage depend on both the first stage and second stage reports. Call a mechanism {\em semi-adaptive} if it depends only on the buyer's declared type. Slightly overloading notation, in such a mechanism the buyer reports a type $t'$ for the first stage, and the seller, based on it, produces a price $p(t')$ for the first stage and a price $q(t')$ for the second (a price can be infinity, in which case the seller does not offer this item). Notice that these mechanisms satisfy a much stronger notion of truthfulness: the buyer is honest in stage two even after a lie in stage one.
This seemingly weaker protocol is optimal.
\begin{lem}
\label{det-semi-adaptive}
There is a revenue-optimal deterministic mechanism that is semi-adaptive.
\end{lem}

Similar results are known for deterministic contracts in sequential screening models. For example,~\cite{courty2000sequential,krahmer2011optimal} show that optimal deterministic contracts can be implemented as a menu of option contracts. The proof Lemma~\ref{det-semi-adaptive} is similar to the proof of these results;  we provide it here for completeness.

\begin{proof}[Proof of Lemma~\ref{det-semi-adaptive}]
Suppose that in a deterministic revenue-optimal mechanism satisfying dynamic incentive compatibility and ex-post individual rationality, 
the price on the second stage $q(t_i, v^{(2)})$ depends on the buyer's types on both stages, $t_i = (v^{(1)}_i, f_i)$ and $v^{(2)}$. 
Fix any first-stage type $t_i = (v^{(1)}_i, f_i)$, 
and let $u^* = \arg \min _{u \geq q(t_i,u)} q(t_i,u)$ be the second stage valuation which minimizes that second stage price, 
among all second-stage valuations for which the item is allocated.

\begin{itemize}[leftmargin=*]
\item $v^{(2)} > q(t_i,u^*)$: the buyer could declare type $u^*$ in order to buy the item for the minimum price. Therefore, since the mechanism is incentive compatible, it must charge $q(t_i,v^{(2)}) = q(t_i,u^*)$.
\item $v^{(2)} < q(t_i,u^*)$: we can assume without loss of generality that the price is again $q(t_i,u^*)$, since the buyer would anyway not buy the item for the current price $q(t_i,v^{(2)}) ( \geq q(t_i,u^*) )$.
\item $v^{(2)} = q(t_i,u^*)$: the buyer's utility remains zero for any price $q(t_i,v^{(2)}) \geq q(t_i,u^*)$; however, the seller's revenue is maximized when selling the item for price $q(t_i,v^{(2)}) = q(t_i,u^*)$.
\end{itemize}

Finally, any buyer with a different first-stage type $t_j = (v^{(1)}_j, f_j)$ that attempts to deviate and declare type $t_i$ on the first stage, would also, with loss of generality, deviate her second-stage valuation to $u^*$.
\end{proof}

Note that it is not clear whether the same is true for randomized mechanisms. Our proof crucially used the fact that, no matter what my true valuation is, the buyer prefers a smaller posted-price. But, in the case of randomized mechanisms, we do not have an order over distributions of prices: one distribution may be more attractive to one type, while another distribution is more attractive for another type. For example, consider a lottery that offers the item for a price of $8$ with probability $1/2$, and with probability $1/2$ doesn't offer the item. This lottery looks more attractive than a posted price of $12$ to a buyer whose value is in the interval $[8,16)$. If the buyer's value is larger than $16$, then the posted price looks more attractive.


\subsection{Simplifying incentive compatibility constraints}
Once we restrict ourselves to semi-adaptive mechanisms, the mechanism becomes two functions $p,q$ mapping the support of the prior to the reals.  Let $p(t)$ be the price charged for the first stage item, and $q(t)$ the price charged for the second stage item, when the buyer reports a type $t$. Let $u(t,t')$ be the expected utility of the buyer when her true type in the first stage is $t$ and she declares $t'$. This utility is the utility of the first stage plus the expected utility for the second stage, when offered a take-it-or-leave-it price $q(t')$. We want $u(t,t) \geq u(t,t')$ for all $t, t' \in V^{(1)}$.

A nice, compact form to express our DIC constraints is using the \emph{reverse} cumulative distribution of the second stage: $\bar{F}_t(x) = Pr[ v^{(2)} \geq x | t ]$. The observation here is that the buyer's second stage utility for a type $t$, when charged price $q$ in stage 2, is $\int_{q}^{\infty} \bar{F}_t(x) dx$. So, for any two possible first-stage types $t = (v^{(1)},\bar{F}_t)$ and report $t' = (b^{(1)}, \bar{F}_{t'})$, the IC constraints are:

\begin{itemize}

\item If both $t$ and $t'$ receive the item on the first stage:
	\[ \int_{q(t')}^{q(t)} \bar{F}_{t'}(x) dx \geq p(t') - p(t) \geq \int_{q(t')}^{q(t)} \bar{F}_{t}(x) dx. \]
\item If neither receives the item on the first stage: $q(t) = q(t')$.
\item If $t'$ receives the item on the first stage, but $t$ does not:

\begin{align*}
\int_{q(t)}^{q(t')} \bar{F}_{t}(x) dx \geq v^{(1)} - p(t'), && b^{(1)} - p(t') \geq \int_{q(t)}^{q(t')} \bar{F}_{t'}(x) dx.
\end{align*}

\end{itemize}

We write $Rev(t_i,p_i,q_i)$ to denote the seller's revenue, when charging the ($i$-th) type $t_i$ the first stage price $p_i$ and second stage price $q_i$. We will write $Rev^{(1)}$ or $Rev^{(2)}$ when we want to refer only to the revenue from the first or second stage, respectively.

%% file: reduction.tex


In this section we briefly describe the construction of our main result.

\begin{thm}
\label{thm:2-day-np-hard}
Finding the optimal deterministic two-stage mechanism is strongly \NP-hard.
\end{thm}

\subsection{Outline}
Given a graph $G = (V,E)$, we construct a joint distribution of valuations such that the optimal feasible revenue (for deterministic DIC and ex-post IR mechanisms) is a strictly increasing function of the maximum independent set in $G$. 

More specifically, with each vertex $i \in G$ we associate a type $t_i$ with valuation $v^{(1)}=B_i$ for the first stage. For each type $t_i$, we want to have two candidate price pairs: $(B_i,C_i)$ or $(A_i,D_i)$. The former will give more revenue, but for every edge $(i,j) \in E$, it will be a violation of the DIC constraints to offer to type $t_i$ the pair $(B_i,C_i)$ and to type $t_j$ the pair $(B_j,C_j)$. Thus, if the difference $r$ in expected revenue between $(B_i,C_i)$ and $(A_i,D_i)$ is the same for all $i$, charging the former for all the vertices of an independent set $S$ and the latter for the rest of the vertices will be a valid pricing, with revenue $\sum_{i\in V} Rev(t_i,A_i,D_i) + r|S|$.

In order to impose the desired structure between $(B_i,C_i)$ and $(A_i,D_i)$, we have an extra type $t^*$, with valuation $v^{(1)}=P^*$ on the first stage. 
$t^*$ appears with very high probability. This way we make most of our revenue from this type, and thus force every revenue-optimal mechanism to charge this type the optimal prices, $(P^*, Q^*)$. The IC constraints for type $t^*$ introduce strong restrictions on the prices for other types. 

The restriction on each edge $(i,j)$ is forced by the IC constraints for $t_i$ and $t_j$, via a careful construction of the distributions over their second-stage valuations. 
The second stage distribution of $t_i$ is $\bar{F}_{t_i}$ and is carefully tuned in the range $[D_{j-1}, D_j]$ (note that $C_j$ is in this interval) depending on whether or not $(i,j) \in E$. 
For ease of notation we write $F_i$ instead of $F_{t_i}$. See Figure~\ref{fig:Fi}.

\begin{figure}[htpb]
\begin{tikzpicture}[scale = 0.75, decoration={markings,
  mark=between positions 0.2 and 1 step 30pt
  with { \draw [fill] (0,0) circle [radius=2pt];}}]

\draw[->] (0,0) -- (16,0) node[anchor=north] {$x$};

\draw	
		(2,0) node[anchor=north] {$C_i$}
		(4,0) node[anchor=north] {$D_i$}
		(7.5,0) node[anchor=north] {$D_{j-1}$}
		(12,0) node[anchor=north] {$C_j$}
		(14,0) node[anchor=north] {$D_j$};
		
\draw	(1,5.2) node{{\scriptsize $\bar{F}_i$ }};
\draw	(1,2.2) node{{\scriptsize $\bar{F}_j$ }};

\draw[->] (0,0) -- (0,6) node[anchor=east] {$Pr[v^{(2)} \geq x]$};

\draw[dotted] (2,0) -- (2,6);
\draw[dotted] (4,0) -- (4,6);
\draw[dotted] (7.5,0) -- (7.5,6);
\draw[dotted] (12,0) -- (12,6);
\draw[dotted] (14,0) -- (14,6);

\draw[very thick] (0,5) -- (2,5) -- (2,4.5) -- (4,4.5) -- (4,4);

\path[postaction={decorate}] (4,3) to (7.5,3);

\draw[very thick,dashed,green] (7.5,4) -- (7.5,3.5) -- (12,3.5) -- (12,2.3) -- (14,2.3) -- (14,1)-- (14.8,1);
\draw[very thick,dotted,red] (7.5,4) -- (7.5,3) -- (12,3) -- (12,2.8) -- (14,2.8) -- (14,1)-- (14.8,1);

\draw[<->] (10,3.5) -- (10,3) node[anchor=east] { };
\draw[<->] (13,2.8) -- (13,2.3) node[anchor=east] { };

\draw[solid] (0pt,0pt) -- (50pt,0pt);

\draw[very thick,blue] (0,2) -- (4,2);
\draw[very thick,blue] (7.5,2) -- (12,2) -- (12,1.5) -- (14,1.5);
\draw[very thick,loosely dashed,blue] (14,1.5) -- (14,1) -- (14.8,1);

\end{tikzpicture}
\caption{\label{fig:Fi}$\bar{F}_i$ when there is (\textcolor{red}{dotted}) an $(i,j)$ edge for $j > i$, and when there isn't (\textcolor{green}{dashed}).}
\end{figure}

\subsection{Construction}

The distribution of valuations on the first stage is rather simple. Let $n=|V|$ denote the number of vertices in $G$. With probability $1-p$, the buyer is of type $t^*$ and has first-stage valuation $v^{(1)} = P^* = n$; with probability $p\cdot w_i$, the buyer is of type $t_i$ and has first-stage valuation $v^{(1)} = B_i = n^2 + 2n + 1 - i$, for $i \in [n]$. The parameters $p$ and $w_i$ are defined later. Notice that the first stage has support of size $n+1$.

We will show that it is always possible to charge type $t_i$ either her full value $B_i$ on the first stage, or slightly less: $A_i = B_i - \epsilon$, for $\epsilon = 1/n^2$. For type $t^*$, we always want to charge the full price, $P^*$. Observe that 
\begin{gather*}
P^* < A_n < B_n < \dots < A_1 < B_1 . 
\end{gather*}

For the second stage we are interested in prices $C_i$ or $D_i$ for $t_i$, and $Q^*$ for $t^*$.
Although we only have $n$ types, 
it will be convenient to think about two more special prices, which we denote $C_{n+1}$ and $D_{n+1}$.
We define $C_i$, $D_i$ and $Q^*$ later; 
for now let us mention that 
\begin{gather*}
C_1 < D_1 < \dots < C_n < D_n < C_{n+1} < D_{n+1} < Q^*. 
\end{gather*}


\subsubsection{Second stage valuations}

The crux of the reduction lies in describing the distributions of the second-stage valuations for each type.
It will be convenient to describe the cumulative distributions $\bar{F}_i(z) = \Pr [ v^{(1)} \geq z | t_i ]$ and $\bar{F}_*(z) = [ v^{(1)} \geq z | t^* ]$. 

The choices of the cumulative distributions in our construction are summarized in Table \ref{table:F_i}, in~\ref{app:reduction_construct}. Type $t_i$ never has nonzero second-stage valuation less than $C_i$, thus the cumulative distribution $\bar{F}_i(x)$ for $x \in (0,C_i)$ is $h_i = \gamma^{-4i}$, for $\gamma = 1+1/n$. Intuitively, this will make $C_i$ an attractive price for the seller. Notice that $\gamma^n \approx e$ is a constant.

At each special price thereafter, $\bar{F}_i$ decreases by some multiplicative factor that is related to $\gamma$.
The exact value of $\bar{F}_i(x)$ for $x \in (D_{j-1}, D_j)$ depends on whether there is an edge $(i,j)$ in $G$.\footnote{For the special prices, $C_{n+1}$ and $D_{n+1}$, assume that all $\bar{F}_i$'s behave as in the ''no edge`` case.}
After $D_{n+1}$, the distribution for all types $t_i$ is the same. $\bar{F}_i$ halves at each $2^k D_{n+1}$, and it is $0$ after  $Q^* = 2^{8\gamma^{4(n+1)}} D_{n+1}$. 

The distribution $\bar{F}_*$ is simpler to describe. $\bar{F}_*(x)$ is equal to $h_1$ for $x \in (0, C_1)$,  and decreases by a multiplicative factor of $\gamma^2$ at each special price thereafter. Type $t^*$ never has valuations between $D_{n+1}$ and $Q^* = 2^{8\gamma^{4(n+1)}} D_{n+1}$.
$\bar{F}_*$ is constant in this domain; in particular $\bar{F}_*(x) = h_* = \frac{A_{n+1} - P^*}{Q^* - D_{n+1}}$. Intuitively, this will make $Q^*$ an attractive price for the seller. Notice also the contrast between this and the gradual decrease of $\bar{F}_i$'s.

We describe how to fix the last parameters in~\ref{app:reduction_construct}. 
We prove the soundness of our construction in~\ref{app:reduction_sound}; the proof of completeness is postponed to~\ref{app:reduction}.

%% file: deterministic.tex

We have three positive results for deterministic mechanisms. We give here a brief sketch of the proofs and postpone further details to the appendix.

Our first result shows that given first stage prices, optimizing over second stage prices (in a way that the joint mechanism is DIC) can be approximated by a fully  polynomial-time approximation scheme (FPTAS), that is, an algorithm that for all $\epsilon > 0$ runs in time polynomial in the size of the input and $\frac{1}{\epsilon}$, and returns a mechanism whose revenue is (at least) a $(1-\epsilon)$ factor of the optimal revenue. For this result, we subdivide the range of second-stage prices into a grid of accuracy $1/K$ (by taking $K$ large enough we obtain an FPTAS), and consider $0-1$ variables who act like indicators for the event ``the price $q_i$ is not larger than the $j$th grid point.''  It turns out that the DIC constraints become totally unimodular. Therefore, the corresponding linear program (which can be solved in polynomial time) has an integral optimum. For more details see~\ref{app:given-first-day}.

\begin{restatable}{theom}{FPTASunimodular}
\label{thm:FPTAS-unimodular}
If the prices in the first stage are fixed, then the optimum deterministic mechanism can be approximated by an FPTAS.
\end{restatable}

Our second result says that if the number of first stage types $|V^{(1)}|$ is a constant, the NP-hardness result no longer holds, and the optimum deterministic mechanism can be computed in polynomial time (polynomial in the support of the second stage distribution, $|V^{(2)}|$).
For this result, we notice that once we have fixed, for each type, the interval between second-stage valuations in which the second-stage price for this type lies (larger than all if the item is not allocated to this type), then the DIC constraints become linear inequalities.  This is because the cumulative distributions are piecewise constant, and thus the integrals in the DIC constraints become linear functions once we know the interval in which the bounds of each integral lie.  Since there are $|V^{(2)}|^{|V^{(1)}|}$ ways to map the $|V^{(1)}|$ second-stage prices to the $|V^{(2)}|$ second-stage intervals, and we assume that $|V^{(1)}|$ is constant, we only need to solve a polynomial number of LPs.

\begin{thm}
\label{thm:constant-support}
If the number of types (the support of first-stage valuations) is constant, then the optimum deterministic mechanism can be computed in polynomial time.
\end{thm}

Our last result states that for independent stages, the optimum deterministic mechanism can be computed in polynomial time.
We observe that once correlation is removed the IC constraints between different types are transitive: satisfied DIC constraints between types $t_i$,$t_j$ and $t_j$,$t_k$ imply satisfied constraints between $t_i$ and $t_k$. Moreover, the allocation function (for the first stage item) is monotone, and the prices of the types that are allocated the first stage item have the following structure. Either the first stage price is equal to the valuation, or the second stage price is zero. Using these observations we can significantly reduce the search space and find the optimal mechanism in polynomial time, essentially by enumerating. For more details see~\ref{app:independent}.

\begin{restatable}{theom}{THMindependent}
\label{thm:independent}
If the stages are independent, the optimum deterministic mechanism can be computed in polynomial time.
\end{restatable}

%% file: LP3.tex

Can we do better by using randomization?  In this section, we construct an LP for the optimum randomized adaptive mechanism for multiple buyers and multiple stages. Specifically, we study the case of $k$ independent buyers and $D$ stages. We show how to compute the optimal randomized mechanism under two definitions of dynamic incentive compatibility (DIC): DIC in dominant strategies (D-DIC) and DIC in Bayesian strategies (B-DIC). We also show that it's possible to compute the optimal mechanism when these notions are satisfied only under honest private histories (``on-path'' truthful), as well as under all private histories (i.e. truthfulness even after a past lie).

The problem we consider here is the natural generalization of the problem in Section~\ref{sec:prelim}.
There are $k$ buyers and $D$ items, sold in $D$ consecutive stages, one item in each stage.
Buyers have additive utility functions, without discounting across time, that is, the final utility of a buyer is the non-discounted sum of her stage utilities.
In stage $d$ buyer $i$ has type/valuation $v^{(d)}_i \in V_i^{(d)}$ that is drawn from a distribution $f_i^{(d)}( v^{(1)}_i, \dots, v^{(d-1)}_i )$. In other words, the type in stage $d$ for buyer $i$ can depend on the realized types from previous stages for buyer $i$ (but is independent of the types of other buyers in this stage or any other stage). The joint distribution is known to the seller, as well as all the buyers, but the realizations are private. That is, the input to the seller's problem is for every buyer $i \in [k]$, and every possible vector of valuations across stages $\vec{v}_i = (v^{(1)}_i,\dots, v^{(D)}_i)$ the probability $Pr[ \vec{v}_i ]$ that this vector of values is realized. We assume that the value zero is in the support of every stage distribution (no matter what the values were for the previous stages), for all buyers. Let $|T_i| = \prod_{j=1}^D |V_i^{(j)}|$ be the number of types of buyer $i$. We note that typically, $|T_i|$ grows exponentially with $D$ (and in what follows we give algorithms that run in time polynomial in the $|T_i|$s). Also, let $V^{[d]} = \times_{i=1}^k V_i^{(d)}$ be the set of possible type vectors in stage $d$.

A dynamic (direct) mechanism consists of, for each buyer $i \in [k]$ and stage $d \in [D]$, an allocation rule $x_i^d$ and a payment rule $p_i^d$ that maps reported values $v^{[d]} = ( v_1^{(d)}, \dots, v_k^{(d)} ) \in V^{[d]}$ to an allocation in $[0,1]$ and a payment in $\mathbb{R}$, respectively. Both of these functions can depend on the public history of reported valuations so far $h^{[d-1]} = (v^{[1]},\dots,v^{[d-1]})$, as well as the outcomes of the mechanism (allocations and payments) $\omega^{[d-1]} = (\omega_1,\dots,\omega_{d-1})$, where $\omega_t = (x,q_1, \dots, q_k)$ if in stage $t$ the item's allocation was $x \in \{0 , 1, \dots, k \}$ and the payment of buyer $i$ was $q_i$.
We therefore write $x_i^d( h^{[d-1]}, \omega^{[d-1]};  v^{[d]} )$ for the allocation and $p_i^d( h^{[d-1]}, \omega^{[d-1]};  v^{[d]} )$ for the payment. Note that the number of outcomes is exponential in the number of stages; as we will see later in this section, without loss of generality we can focus on mechanisms whose allocation and payment rules in stage $d$ do not depend on $\omega^{[d-1]}$.

The order of events in stage $d$ is as follows: 
\begin{enumerate}[label=(\roman*)]
\item Each buyer $i$ privately learns her type $v_i^{(d)}$ (which is drawn from a distribution that is possibly correlated with her true past types).
\item Each buyer $i$ reports a type $\hat{v}_i^{(d)}$ to the mechanism.
\item The mechanism allocates the item to buyer $i$ with probability $x_i^d( h^{[d-1]}, \omega^{[d-1]};  \hat{v}^{[d]} )$ and charges $p_i^d( h^{[d-1]}, \omega^{[d-1]};  \hat{v}^{[d]} )$ (and note that the allocation and payment could be correlated). 
\item Each buyer $i$ obtains (realized) stage utility $u^{(d)}_i = v_i^{(d)} - p_i^d$ if they were allocated the item (and $u^{(d)}_i =  - p_i^d( . )$ if they were not allocated but were charged a payment).
\item The public history so-far is updated to include the reported types  $\hat{v}^{[d]}$, and the outcome in stage $d$ (which buyer got the item  and payments).
\end{enumerate}

We note that, importantly, when the buyer decides what to report in step (ii), they evaluate the stage utility from stage $d$ in expectation over the realization of the stage $d$ lottery, i.e. they calculate $u^{(d)}_i = v_i^{(d)} \cdot x_i^d (.; \hat{v}^{[d]}) - p_i^d(.;\hat{v}^{[d]})$.

\paragraph{Feasibility} A mechanism is feasible if it allocates the item in stage $d$ to at most one buyer in expectation. That is, for all $d \in [D]$, for all histories $h^{[d-1]}$, $\omega^{[d-1]}$ and possible reports $v^{[d]}$, $\sum_{i \in [k]} x_i^d( h^{[d-1]}, \omega^{[d-1]};  v^{[d]} ) \leq 1$.

\paragraph{Incentive compatibility} 
The dynamic revelation principle~\cite{myerson1986multistage,sugaya2021revelation} states that there is no loss of generality in restricting attention to dynamic direct mechanisms where buyers report their information truthfully ``on-path'', that is, the mechanism does not need to guarantee truth-telling after a lie in a previous stage. This distinction was not important for the deterministic two stage mechanism of Section~\ref{sec:prelim}, since optimal mechanisms for the weaker notion ended up being semi-adaptive (which satisfies the stronger notion). However, this might not be the case for the multi-agent, multi-stage problem we study here (since, e.g., the future utility of a buyer in stage $2$ is calculated using both the public history and the private history, while her stage utility only depends on the public history and the current value). It is intuitively clear that asking for truth-telling even after lies is a strictly harder computational task, since it needs to take care of multi-stage deviations (and there are exponentially many of these, even for two stages). We show that, in our setting, optimizing with respect to mechanisms that satisfy either notion of incentive compatibility is computationally feasible. Roughly speaking, optimizing with respect to the stronger notion will boil down to writing a linear program with an IC constraint for each buyer and each public history and private history (which is a polynomial number of constraints with respect to the size of the input/the joint distribution over types). Optimizing with respect to the weaker notion boils down to writing an IC constraint for each buyer and each public history (since the assumption here is that a buyer has been honest so far, her private history matches the public history), which is a smaller number of constraints. For the remainder of this section we will focus on the computationally harder task of computing the optimal mechanism that satisfies the stronger notion of incentive compatibility; we explain how our solution can be adjusted to work for the weaker notion after we prove our main theorem.

A mechanism is dynamic incentive compatible (DIC) if, informally, each buyer $i$ is better off reporting her true value $v^{(d)}_i$ in each stage $d$. We consider two variations of DIC. The first one, DIC in dominant strategies (D-DIC) requires that truthful reporting in stage $d$ maximizes the utility of each buyer, regardless of the past (i.e. for all histories), and regardless other buyers' reports.
The second one, DIC in Bayesian strategies (B-DIC) requires that truthful reporting in stage $d$ maximizes the utility of each buyer, regardless of the past, and in expectation over the other buyers' truthful reports in this stage and the future.
We note that often D-DIC is too strict of a requirement in general dynamic mechanism design settings and might come with loss of generality (see for example~\cite{bergemann2010dynamic,bergemann2019dynamic}). However, in our setting, natural mechanisms satisfy this seemingly strong requirement. For example, running a second-price auction in each stage, which is also welfare optimal, satisfies this notion. We therefore suggest that, even though this class might be with loss of generality, it is worth optimizing over.

We start by defining D-DIC.
Naively, since re-reporting of past types is not allowed, the incentive constraints must account for multi-shot deviations. 
For example, lying only in stage one could decrease a buyer's utility, and so does lying only in stage two; and yet lying on both stages increases her expected utility.  Therefore, when deciding to lie, the buyer must choose in advance among all exponentially many different strategies that deviate from the truth now and in the future. Notice that the number is exponential even for two stages: the optimization is over functions from true pairs of types to declared pairs of types. This is the first obstacle in writing a polynomial size LP. However,  similarly to Section~\ref{sec:prelim}, we can define both variations of DIC via backward induction. 
Specifically, for the last stage $D$, for all buyers $i \in [k]$, for all histories of reported types $h^{[D-1]} = (b^{[1]},\dots,b^{[D-1]})$, and past outcomes of the mechanism $\omega^{[D-1]} = (\omega_1,\dots,\omega_{D-1})$, and for all last stage reports $v^{(D)}_{-i}$ of the other buyers, a D-DIC mechanism guarantees that the (expected over the randomness of the mechanism) stage utility of $i$ in stage $D$ when reporting her true value $v^{(D)}_i$ is at least her utility from reporting a different value $\hat{v}$:
\begin{multline*}
v^{(D)}_i \cdot x^D_i ( h^{[D-1]}, \omega^{[D-1]};  v^{(D)}_{-i}, v^{(D)}_i ) - p^D_i ( h^{[D-1]}, \omega^{[D-1]};  v^{(D)}_{-i}, v^{(D)}_i ) \geq \\
v^{(D)}_i \cdot x^D_i ( h^{[D-1]}, \omega^{[D-1]};  v^{(D)}_{-i}, \hat{v} ) - p^D_i ( h^{[D-1]}, \omega^{[D-1]};  v^{(D)}_{-i}, \hat{v} ).
\end{multline*}

\noindent Let $U_i^{D}(h^{[D-1]}, \mathrm{\bar{h}}_i^{[D-1]},\omega^{[D-1]},v^{(D)}_{-i})$ be the expected utility of buyer $i$ for participating in the last stage $D$, when the public history is $h^{[D-1]},\omega^{[D-1]}$, the private history of buyer $i$ is $\mathrm{\bar{h}}_i^{[D-1]} = (v^{(1)}_i,\dots,v^{(D-1)}_i)$, and the reports of other players in stage $D$ are $v^{(D)}_{-i}$. By the previous argument, this utility is calculated assuming truthful reporting for buyer $i$ in stage $D$:
\[
U_i^{D}(h^{[D-1]}, \mathrm{\bar{h}}_i^{[D-1]},\omega^{[D-1]},v^{(D)}_{-i}) = \mathbb{E}_{v^{(D)}_i | \mathrm{\bar{h}}_i^{[D-1]} }\left[ v^{(D)}_i \cdot x^D_i ( h^{[D-1]}, \omega^{[D-1]};  v^{[D]}) - p^D_i ( h^{[D-1]}, \omega^{[D-1]};  v^{[D]} ) \right].
\] 
We note that $U_i^{D}(h^{[D-1]}, \mathrm{\bar{h}}_i^{[D-1]},\omega^{[D-1]},v^{(D)}_{-i})$ (and generally $U_i^{d}(.)$) for ``on-path'' truthfulness, i.e. when the mechanism doesn't require truth after lies, is defined slightly differently, by taking an additional maximum over the report in stage $D$, i.e. it is equal to 
\[
\mathbb{E}_{v^{(D)}_i | \mathrm{\bar{h}}_i^{[D-1]} }\left[ \max_{\hat{v}\\
} v^{(D)}_i \cdot x^D_i ( h^{[D-1]}, \omega^{[D-1]};  v^{(D)}_{-i}, \hat{v} ) - p^D_i ( h^{[D-1]}, \omega^{[D-1]};  v^{(D)}_{-i}, \hat{v}) \right]
\].

For the second to last stage, let $u_i^{D-1}(h^{[D-2]},\omega^{[D-2]}; v^{(D-1)}_{-i}, v^{(D-1)}_i \rightarrow \hat{v})$ be the \emph{stage} utility of buyer $i$ when the history is $h^{[D-2]}, \omega^{[D-2]}$, all other buyers report $v^{(D-1)}_{-i}$, her true value in stage $D-1$ is $v^{(D-1)}_i$, but she reports $\hat{v}$. We note that the stage utility does not depend on the buyer's private history $\mathrm{\bar{h}}_i^{[D-2]}$. Let $h^{[D-1]} = [h^{[D-2]}, v^{(D-1)}_{-i}, v^{(D-1)}_i]$ be the public history (in stage $D$) if $i$ reports truthfully in stage $D-1$, and $\hat{h}^{[D-1]} = [h^{[D-2]}, v^{(D-1)}_{-i}, \hat{v}]$ be the public history (in stage $D$) if $i$ reports $\hat{v}$ in stage $D-1$. The private history in stage $D-1$ is $\mathrm{\bar{h}}_i^{[D-1]} = [\mathrm{\bar{h}}_i^{[D-2]}, v^{(D-1)}_i]$ and is independent of the report.

Then, for  D-DIC we have that for all $h^{[D-2]},\mathrm{\bar{h}}_i^{[D-2]},\omega^{[D-2]}, v^{(D-1)}_{-i}, v^{(D)}_{-i}, v^{(D-1)}_{i}$ and $\hat{v}$:
\begin{multline*}
u^{D-1}_i(h^{[D-2]},\omega^{[D-2]}; v^{(D-1)}_{-i}, v^{(D-1)}_i \rightarrow v^{(D-1)}_i) + \mathbb{E}_{\omega_{D-1}}[ U_i^{D}( h^{[D-1]},\mathrm{\bar{h}}_i^{[D-1]}, [\omega^{[D-2]}, \omega_{D-1}], v^{(D)}_{-i} )  ]  \geq \\ 
u^{D-1}_i(h^{[D-2]},\omega^{[D-2]}; v^{(D-1)}_{-i}, v^{(D-1)}_i \rightarrow \hat{v} ) + \mathbb{E}_{\omega_{D-1}}[ U_i^{D}( \hat{h}^{[D-1]},\mathrm{\bar{h}}_i^{[D-1]}, [\omega^{[D-2]}, \omega_{D-1}],v^{(D)}_{-i} )  ].
\end{multline*}
Note that the definition of $U_i^{D}$ includes the expectation over types in stage $D$, therefore in the second term (which captures the remaining utility from participating in the mechanism) we only need to take an expectation over the outcome of the mechanism in stage $D-1$.

Proceeding backwards for all stages $d$, for D-DIC we have that for every buyer $i$, all histories so-far $h^{[d-1]},\mathrm{\bar{h}}_i^{[d-1]},\omega^{[d-1]}$, all reports for the other buyers in stage $d$ ($v^{(d)}_{-i}$) and stage $d+1$ onward ($v^{[d+1]:[D]}_{-i}$), all types of buyer $i$ $v^{(d)}_{i}$ and all stage $d$ deviations $\hat{v}$:
\begin{multline*}
u^{d}_i(h^{[d-1]},\omega^{[d-1]}; v^{(d)}_{-i}, v^{(d)}_i \rightarrow v^{(d)}_i) + \mathbb{E}_{\omega_{d}}[ U_i^{d+1}( h^{[d]},\mathrm{\bar{h}}_i^{[d]}, [\omega^{[d-1]}, \omega_{d}], v^{[d+1]:[D]}_{-i} )  ]  \geq \\ 
u^{d}_i(h^{[d-1]},\omega^{[d-1]}; v^{(d)}_{-i}, v^{(d)}_i \rightarrow \hat{v} ) + \mathbb{E}_{\omega_{d}}[ U_i^{d+1}( \hat{h}^{[d]},\mathrm{\bar{h}}_i^{[d]}, [\omega^{[d-1]}, \omega_{d}],v^{[d+1]:[D]}_{-i} )  ],
\end{multline*}
where (1) $u^{d}_i(h^{[d-1]},\omega^{[d-1]}; v^{(d)}_{-i}, v^{(d)}_i \rightarrow \hat{v})$ is the stage utility of buyer $i$ when the history is $h^{[d-1]}, \omega^{[d-1]}$, all other buyers report $v^{(d)}_{-i}$, her true value in stage $d$ is $v^{(d)}_i$, but she reports $\hat{v}$, (2) $h^{[d]}$ and $\hat{h}^{[d]}$ are the two public histories in stage $d+1$ that correspond to truthful and non-truthful reporting of buyer $i$ in stage $d$, and (3) $U_i^{d+1}( h^{[d]},\mathrm{\bar{h}}_i^{[d]}, \omega^{[d]}, v^{[d+1]:[D]}_{-i} )$ is the expected utility of buyer $i$ for participating in the mechanism in stages $d+1$ through $D$ (which depends on the private history $\mathrm{\bar{h}}_i^{[d-1]}$ and the other buyers' future reports $v^{[d+1]:[D]}_{-i} = (v^{d+1}_{-i}, v^{d+2}_{-i}, \dots )$). The term $U_i^{d+1}( h^{[d]},\mathrm{\bar{h}}_i^{[d]}, \omega^{[d]},v^{[d+1]:[D]}_{-i} )$ is typically referred to as the continuation utility, and is equal to:
\[
\mathbb{E}_{v^{[d+1]:[D]}_i, \omega_{d+1:D} | \mathrm{\bar{h}}_i^{[d]} } \left[ \sum_{t = d+1}^D  u^{t}_i(  [ h^{[d]}, v^{[d+1]:[t-1]}], [ \omega^{[d]}, \omega_{d+1:t-1} ]; v^{(t)}_{-i}, v^{(t)}_i \rightarrow v^{(t)}_i) \right],
\]
where $v^{[z]:[\ell]} = v^{[z]}, \dots, v^{[\ell]}$, and $\omega_{z: \ell} = \omega_{z}, \dots, \omega_{\ell}$.
Equivalently, $U_i^{d+1}( h^{[d]},\mathrm{\bar{h}}_i^{[d]}, \omega^{[d]},v^{[d+1]:[D]}_{-i}  )$ can be written recursively as
\begin{multline*}
\mathbb{E}_{v^{(d+1)}_i, \omega_{d+1} | \mathrm{\bar{h}}_i^{[d]} } \left[ u^{t}_i(  [ h^{[d]}, v^{[d+1]}], [ \omega^{[d]}, \omega_{d+1} ]; v^{(d+1)}_{-i}, v^{(d+1)}_i \rightarrow v^{(d+1)}_i) \right. \\
\left. + U_i^{d+2}( [h^{[d]}, v^{[d+1]}], [\mathrm{\bar{h}}_i^{[d]}, v^{[d+1]}], [\omega^{[d]}, \omega_{d+1} ], v^{[d+2]:[D]}_{-i} ) \right].
\end{multline*}

For ``on-path'' truthfulness, we need to take an additional maximum over the report of stage $d+1$ inside this expectation, i.e.  $U_i^{d+1}( h^{[d]},\mathrm{\bar{h}}_i^{[d]}, \omega^{[d]},v^{[d+1]:[D]}_{-i}  )$ is equal to
\begin{multline*}
\mathbb{E}_{v^{(d+1)}_i, \omega_{d+1} | \mathrm{\bar{h}}_i^{[d]} } \left[ \max_{\hat{v}} u^{t}_i(  [ h^{[d]}, v^{[d+1]}], [ \omega^{[d]}, \omega_{d+1} ]; v^{(d+1)}_{-i}, v^{(d+1)}_i \rightarrow \hat{v}) \right. \\
\left. + U_i^{d+2}( [h^{[d]}, [v^{[d+1]}_{-i}, \hat{v}] ], [\mathrm{\bar{h}}_i^{[d]}, v^{[d+1]}], [\omega^{[d]}, \omega_{d+1} ],v^{[d+2]:[D]}_{-i} ) \right].
\end{multline*}

For B-DIC we need to take an additional expectation over $v^{[d+1]:[D]}_{-i}$, where each $v^{(t)}_{j}$ for buyer $j$ is drawn from the marginal conditioned on $h^{[t-1]}_{j}$, the history of reports of buyer $j$. Note that this argument assumes that all other buyers have been truthful so far. Using identical arguments as above, we have that for every buyer $i$, all histories so-far $h^{[d-1]},\mathrm{\bar{h}}_i^{[d]},\omega^{[d-1]}$, all $v^{(d)}_{i}$, and all stage $d$ deviations $\hat{v}$:
\begin{multline*}
\mathbb{E}_{v^{(d)}_{-i}} \left[
u^{d}_i(h^{[d-1]},\omega^{[d-1]}; v^{(d)}_{-i}, v^{(d)}_i \rightarrow v^{(d)}_i) + \mathbb{E}_{\omega_{d}|\mathrm{\bar{h}}_i^{[d-1]}}[ U_i^{d+1}( h^{[d]},\mathrm{\bar{h}}_i^{[d]}, [\omega^{[d-1]}, \omega_{d}] )  ] \right]  \geq \\ 
\mathbb{E}_{v^{(d)}_{-i}} \left[
u^{d}_i(h^{[d-1]},\omega^{[d-1]}; v^{(d)}_{-i}, v^{(d)}_i \rightarrow \hat{v} ) + \mathbb{E}_{\omega_{d}|\mathrm{\bar{h}}_i^{[d-1]}}[ U_i^{d+1}( \hat{h}^{[d]},\mathrm{\bar{h}}_i^{[d]} [\omega^{[d-1]}, \omega_{d}] )  ] \right].
\end{multline*}

\paragraph{Individual rationality} A mechanism is ex-post individual rational (ex-post IR), if once all uncertainty is resolved, every buyer always has non-negative utility. Given a history of $h^{[d-1]}, \mathrm{\bar{h}}_i^{[d-1]}, \omega^{[d-1]}$ in stage $d$, let $\hat{u}_i^d( h^{[d-1]}, \mathrm{\bar{h}}_i^{[d-1]}, \omega^{[d-1]} )$ be the utility buyer $i$ has accumulated so far, i.e. the sum of values for the items she received minus the payments, where the allocations and payments are according to $\omega^{[d-1]}$, and values are according to $h^{[d-1]}$ and $\mathrm{\bar{h}}_i^{[d]}$.
A dynamic mechanism outputs for every stage $d \in [D]$, history $h^{[d-1]}, \omega^{[d-1]}$ and stage $d$ reports $v^{[d]}$ a distribution over stage $d$ outcomes: an outcome $\omega_d = ( x, q_1, \dots, q_k )$ occurs with probability $Pr[\omega_d]$ (that depends on $h^{[d-1]}, \omega^{[d-1]}$ and $v^{[d]}$) specifies which buyer $x \in \{ 0, 1, \dots, k \}$ won the item (if any) and the payment $q_i$ for each buyer $i$. The exact distribution depends on the correlation between $x_i^d(.)$ and $p_i^d(.)$. Formally, an ex-post individual rational mechanism satisfies, for all stages $d \in [D]$, all public histories $h^{[d-1]}, \omega^{[d-1]}$ and private histories $\mathrm{\bar{h}}_i^{[d-1]}$ that match the public history (i.e. we only guarantee individual rationality if the reports so far have been honest), stage $d$ reports $v^{[d]}$ and stage $d$ outcomes $\omega_d = ( x, q_1, \dots, q_k )$ that occur with positive probability:
\[
\hat{u}_i^d( h^{[d-1]}, \mathrm{\bar{h}}_i^{[d-1]}, \omega^{[d-1]} ) - q_i + v^{(d)}_i \cdot I ( x = i ) \geq 0,
\]
where $I( e )$ is the standard indicator function for an event $e$ (takes the value $1$ if $e$ occurs, and zero otherwise).
Since we guarantee ex-post IR with respect to the public history (that is, we provide no guarantees for buyers that have lied), we will omit the private history of a buyer when talking about the IR constraints, and assume it matches the public history.

Before we give our main result we make a number of simplifications to the expression above for ex-post IR.

\paragraph{Stage-wise ex-post IR versus ex-post IR}
The reason an adaptive mechanism keeps track of the allocations and payments so far is so that is can ensure that the ex-post IR constraint is satisfied. 
Of course, having a variable for every possible past outcome is prohibitively costly: even when the number of types is small, the number of possible outcomes grows exponentially with the number of stages. Fortunately, one can show that without loss of generality, we can focus on stage-wise ex-post IR mechanisms, that is, non-negative stage utility. Formally, a mechanism is stage-wise ex-post IR if for all stages $d \in [D]$, public history $h^{[d-1]}$ and private history $\mathrm{\bar{h}}_i^{[d-1]}$ that matches the public history, stage $d$ reports $v^{[d]}$ and stage $d$ all outcomes $\omega_d = ( x, q_1, \dots, q_k )$ that occur with positive probability (which depends on the history and reports):
\[
v^{(d)}_i \cdot I ( x = i ) - q_i \geq 0.
\]
Equivalently, if buyer $i$ does not receive the item she should not pay, and if she receives the item she should pay at most her (reported) value.

The equivalence between ex-post IR and stage-wise ex-post IR can be shown via an easy reduction (from ex-post IR to stage-wise ex-post IR), see for example~\cite{mirrokni2016optimal}. For completeness we include this reduction in~\ref{app:stage-wise}. The main idea is that given an ex-post IR mechanism $M$, one can construct a stage-wise ex-post IR mechanism $M'$ that charges each bidder exactly their bid for each of the first $D-1$ stages, and in the last stage charges the difference between the total payments in $M$ and the total payments so far in $M'$. Therefore, we can restrict ourselves to mechanisms that guarantee ex-post non-negative utility in each stage, and that only take as input, in each stage, the history of reported values so far.

\paragraph{Correlation between allocation and payment}
A final issue when writing a linear program is the choice of variables. The ``standard'' way, by which we mean the most common formulation for one-shot mechanisms, is to have, for each buyer and each vector of reports, a variable for the expected allocation and a variable for the expected payment. 
Taking this to the dynamic setting, we would have for each stage, history, buyer and reports a variable for the expected allocation and a variable for the expected payment.
But, because of the ex-post IR constraint when implementing the corresponding mechanism, correlation is necessary. To see this most clearly, consider an example where the LP, in a certain situation (i.e. stage, history etc) allocated an item to buyer $i$ with probability $1/2$ and the expected payment was $p$. How would we implement this in an ex-post IR way, so that the buyer's utility is always non-negative? For example, we would need to ensure that when the item is not allocated the payment is zero. Thus, correlation between payment and allocation is necessary, a feature that seems difficult to work if the goal is to write a poly-sized linear program.\footnote{For example, a natural way to encode this correlation is to have a variable for the probability of each possible (allocation,payment) outcome, which leads to exponentially many variables.}  Fortunately, a simple correlation scheme will allow us to by-pass this issue.

Consider a feasible, stage-wise ex-post IR, and DIC (D-DIC or B-DIC) mechanism given via allocation probabilities and expected payments. That is, for every stage $d$, histories $h^{[d]}, \omega^{[d]}$, buyer $i$, and possible reports $v^{[d]}$,  $x_i^{d}(h^{[d]},\omega^{[d]};v^{[d]})$ is the probability that this buyer is allocated item $d$, and $p_i^{d}(h^{[d]},\omega^{[d]};v^{[d]})$ is the expected payment. Stage-wise ex-post individual rationality implies that $v_i \cdot x_i^{d}(h^{[d]},\omega^{[d]};v^{[d]}) - p_i^{d}(h^{[d]},\omega^{[d]};v^{[d]}) \geq 0$. We can implement this mechanism in a way that the ex-post IR constraint is respected even after the random outcome for stage $d$ is selected (feasibility and truthfulness will be immediately implied). With probability $x_i^{d}(h^{[d]},\omega^{[d]};v^{[d]})$ we allocate the item to buyer $i$ and charge her $p_i^{d}(h^{[d]},\omega^{[d]};v^{[d]}) / x_i^{d}(h^{[d]},\omega^{[d]};v^{[d]})$. With probability $1-x_i^{d}(h^{[d]},\omega^{[d]};v^{[d]})$ we do not allocate to $i$, and charge her nothing. The expected utility of the buyer remains $v_i \cdot x_i^{d}(h^{[d]},\omega^{[d]};v^{[d]}) - p_i^{d}(h^{[d]},\omega^{[d]};v^{[d]})$. Furthermore, when the item is allocated, her utility is $v_i - p_i^{d}(h^{[d]},\omega^{[d]};v^{[d]})/x_i^{d}(h^{[d]},\omega^{[d]};v^{[d]}) \geq 0$, and zero when not allocated, so the stage-wise ex-post IR constraint is satisfied.

\subsection*{The optimal randomized mechanism LP}

Finally, we are ready to describe our LP for computing randomized adaptive mechanism for $k>1$ buyers and $D>2$ stages. 
Recall that the input to our problem is, for every buyer $i \in [k]$, and every possible vector of valuations across stages $\vec{v}_i = (v^{(1)}_i,\dots, v^{(D)}_i)$ the probability $Pr[ \vec{v}_i ]$ that this vector of values is realized. Let $|T_i| = \prod_{j=1}^D |V_i^{(j)}|$ be the number of types of buyer $i$. Note that typically, $|T_i|$ grows exponentially with $D$. We use $|T| = \sum_{i=1}^k |T_i|$ for the total number of types. We also use $|V| = max_{i,d} |V^{(d)}_i|$ for the support of the ``largest'' marginal distribution for any stage $d$ and buyer $i$.

We prove the D-DIC case first and discuss how to alter the proof to take care of the B-DIC case, and ``on-path'' truthfulness, afterward.


\begin{restatable}{theom}{lpmanydaysmanybidders}
\label{thm:lp-manydays-manybidders}
For any number of stages $D$, and a constant number of independent buyers $k$, the optimal, adaptive, randomized D-DIC mechanism can be found in time $poly(D,|T|^{2k+3})$.
\end{restatable}

\begin{proof}
Our LP has a variable $x_i^{d}( h^{[d-1]}; v^{[d]} )$ and $p_i^{d}( h^{[d-1]}; v^{[d]} )$ for the probability of allocating item $d$ and the payment, respectively, to buyer $i$ when the reports on stage $d$ are according to $v^{[d]} = ( v^{(d)}_1, \dots, v^{(d)}_k )$ (where $v^{(d)}_j$ is the report of buyer $j$), and the history up until stage $d$ is according to $h^{[d-1]} = ( v^{(1)}_1, \dots, v^{(d-1)}_k )$. The number of variables is therefore $O( k \cdot  D \cdot |T|^k \cdot |V|^k ) = poly(D ,|T|^{k+1} )$.

\begin{itemize}
\item \textbf{Objective}: Our objective is to maximize expected revenue, which can be expressed in our variables as
\[
R = \sum_{d=1}^D \sum_{h^{[d-1]}}  \sum_{v^{[d]}} Pr[ h^{[d-1]}, v^{[d]} ] \sum_{i=1}^k p_i^{d}( h^{[d-1]}; v^{[d]} ).
\]
Notice is that $ Pr[ h^{[d-1]}, v^{[d]} ]$ can be easily computed from our input. It is the product over buyers $i$ of $Pr[ v^{(1)}_i, \dots, v^{(d)}_i ]$; the latter term can be calculated by summing up the probabilities of all possible (at most $|T|$ of them) futures for buyer $i$.
\item \textbf{Feasibility}: We need to ensure that we are not over-allocating any item, i.e. allocating an item with probability more than $1$.
\[
\forall d \in [D], h^{[d-1]}, v^{[d]}: \sum_{i=1}^k x_i^{d}( h^{[d-1]}; v^{[d]} ) \leq 1.
\]
The number of constraints is $O( D \cdot |T|^k \cdot |V|^k) = O(D|T|^{k+1})$.
\item \textbf{Stage-wise ex-post IR}: The stage utility should be non-negative for any buyer $i$, all stages $d$, histories $h^{[d-1]}$, valuations $v^{[d]}$:
\[
\forall i \in [k], d \in [D], h^{[d-1]}, v^{[d]}: v^{(d)}_i x_i^{d}( h^{[d-1]}; v^{[d]} ) - p_i^{d}( h^{[d-1]}; v^{[d]} ) \geq 0.
\]
The number of constraints is $O( k \cdot D \cdot |T|^k \cdot |V|^k) = O(kD|T|^{k+1})$.
\item \textbf{Incentive compatibility}: In order to express D-DIC compactly, we introduce the following intermediate variables. Let $U_i(d+1; h^{[d-1]},\mathrm{\bar{h}}_i^{[d-1]}, v^{(d)}_i, v^{[d]:[D]}_{-i}, t)$ be the expected utility of buyer $i$ from participating in the mechanisms in stages $d+1$ through $D$, when the public history up until stage $d$ is $h^{[d-1]}$, buyer $i$'s private history is $\mathrm{\bar{h}}_i^{[d-1]}$,  the stage $d$ report for buyer $i$ is $v^{(d)}_i$, the other agents' stage $d$ and future reports are according to  $v^{[d]:[D]}_{-i}$, and buyer $i$'s true stage $d$ type is $t$. The number of these variables is $O( k \cdot D \cdot |T|^k \cdot |T| \cdot |V| \cdot |T|^{k-1} \cdot |V| )$, thus the total number of variables remains polynomial in the size of the input for a constant $k$. Furthermore, as we've already discussed, the future expected utility can be easily calculated. Thus, we can recursively define $U_i(d+1; h^{[d-1]},\mathrm{\bar{h}}_i^{[d-1]}, v^{(d)}_i, v^{[d]:[D]}_{-i},t)$ as follows (and we define $U_i(D+1;.)= U_i(D+2;.) = 0$ for all $i$ to make the recursion easier to write):
\begin{multline*}
U_i(d+1; h^{[d-1]},\mathrm{\bar{h}}_i^{[d-1]}, v^{(d)}_i, v^{[d]:[D]}_{-i},t) = \sum_{v^{(d+1)}_i } Pr[v^{(d+1)}_i | \mathrm{\bar{h}}_i^{[d-1]},t] \cdot \\ 
\left( v^{(d+1)}_i x_i^{d+1}( [h^{[d-1]}, v^{(d)}_i, v^{(d)}_{-i}]; [v^{(d+1)}_i, v^{(d+1)}_{-i}] ) - p_i^{d+1}([h^{[d-1]}, v^{(d)}_i, v^{(d)}_{-i}]; [v^{(d+1)}_i, v^{(d+1)}_{-i}] )  \right. \\ \left. + U_i(d+2; [h^{[d-1]}, v^{(d)}_i, v^{(d)}_{-i}],\mathrm{\bar{h}}_i^{[d]}, v^{(d+1)}_i, v^{[d+1]:[D]}_{-i}, v^{(d+1)}_i) \right).
\end{multline*}

Writing our DIC constraints is now much simpler. Specifically, for the case of D-DIC, we have that for all stages $d \in [D]$, buyers $i \in [k]$, histories up until stage $d$ $h^{[d-1]}$ and $\mathrm{\bar{h}}_i^{[d-1]}$, possible true values $v$, misreports $\hat{v}$ for buyer $i$, and all possible stage $d$ and future valuations for the remaining buyers $v^{[d]:[D]}_{-i}$:
\begin{multline*}
v \cdot x_i^{d}(h^{[d-1]}; [v, v^{(d)}_{-i}]) - p_i^{d}(h^{[d-1]}; [v, v^{(d)}_{-i}])  + U_i(d+1; h^{[d-1]},\mathrm{\bar{h}}_i^{[d-1]}, v, v^{[d]:[D]}_{-i},v) \geq \\
v \cdot x_i^{d}(h^{[d-1]}; [\hat{v}, v^{(d)}_{-i}]) - p_i^{d}(h^{[d-1]}; [\hat{v}, v^{(d)}_{-i}])   + U_i(d+1; h^{[d-1]},\mathrm{\bar{h}}_i^{[d-1]}, \hat{v}, v^{[d]:[D]}_{-i},v).
\end{multline*}

The number of constraints is $O( D \cdot k \cdot |T|^{k+2} \cdot |V|^2 \cdot |T|^{k-1}) = O( D k |T|^{2k+3} )$.
\end{itemize}

Given a solution to this LP, we can implement a mechanism as follows. On the first stage the mechanism elicits reports $v^{[1]}$ for the first stage item, and allocates the item to buyer $i$ with probability $x_i^1(v^{[1]})$; if the item is allocated to $i$, her payment is $p_i^1(v^{[1]})/x_i^1(v^{[1]})$, which as we've already argued is non-negative, and if the item is not allocated, the payment is zero. In the second stage reports $v^{[2]}$ are submitted to the mechanism, and the mechanism allocates to buyer $i$ with probability $x_i^2(v^{[1]};v^{[2]})$; if the item is allocated the payment is $p_i^2(v^{[1]};v^{[2]})/x_i^2(v^{[1]};v^{[2]})$, and so on. The probability of allocating an item is at most $1$ by the feasibility constraints. Similarly, incentive compatibility is satisfied via a similar argument.

We note that we do not have variables for histories outside the support, but we can easily handle such deviations (while maintaining feasibility, truthfulness and ex-post IR) as follows. First solve the LP as is to get a mechanism $M$; then, at runtime, if the report of buyer $i$ at some stage is outside the support, ``pretend''  that the buyer reported the lowest value in her support. Importantly, we now have a valid history as an input to the future allocation variables. Call this mechanism $M'$; notice that $M'$ remains truthful (deviating outside the support is the same as deviating to the lowest value in the support), and ex-post IR (if the buyer does get the item, the payment is at most the lowest value, which is a lower bound on her real value).

Given a D-DIC and ex-post IR mechanism $M$ we now show that there is a corresponding feasible solution to our LP.
In full generality, in stage $d$ elicits reports $v^{[d]}$, and depending on the history and outcomes $h^{[d]}, \omega^{[d]}$ so far draws an (allocation,payment) pair from some set $\Theta_{d,h^{[d]}, \omega^{[d]},v^{[d]}}$, i.e. with some probability $Pr[\theta]$ outputs the outcome $\theta \in \Theta_{d,h^{[d]}, \omega^{[d]},v^{[d]}}$ that allocated the item to buyer $i$ and charges her $p \leq v^{(d)}_i$. 
We can then easily write variables that are a feasible solution to the above LP: $x_i^d(h^{[d]}: v^{[d]}) = \sum_{\theta \in \Theta_{d,h^{[d]}, \omega^{[d]},v^{[d]}} : \theta = (i,.)} Pr[\theta]$, and $p_i^d(h^{[d]}: v^{[d]}) = \sum_{\theta \in \Theta_{d,h^{[d]}, \omega^{[d]},v^{[d]}} : \theta = (i,p)} p \cdot Pr[\theta]$. It is easy to check that all the LP constraints are satisfied, and the revenue of $M$ is exactly the value of $R$ for the feasible solution we described.
\end{proof}

For B-DIC we have the exact same statement:
\begin{thm}
For any number of stages $D$, and a constant number of independent buyers $k$, the optimal, adaptive, randomized B-DIC mechanism can be found in time $poly(D,|T|^{2k+3})$.
\end{thm}

The only difference in the LP itself is the incentive compatibility constraints. For D-DIC we had
\begin{multline*}
v x_i^{d}(h^{[d-1]}; [v, v^{(d)}_{-i}]) - p_i^{d}(h^{[d-1]}; [v, v^{(d)}_{-i}])  + U_i(d+1; h^{[d-1]},\mathrm{\bar{h}}_i^{[d-1]}, v, v^{[d]:[D]}_{-i},v) \geq \\
v x_i^{d}(h^{[d-1]}; [\hat{v}, v^{(d)}_{-i}]) - p_i^{d}(h^{[d-1]}; [\hat{v}, v^{(d)}_{-i}])   + U_i(d+1; h^{[d-1]},\mathrm{\bar{h}}_i^{[d-1]}, \hat{v}, v^{[d]:[D]}_{-i},v).
\end{multline*}

For B-DIC, we simply take an expectation over $v^{[d]:[D]}_{-i}$, an operation that we can afford computationally (since it's a sum of  $O(|T|^k)$ terms).
The proof of correctness remains virtually unchanged (noting that none of the arguments used D-DIC in a non-trivial way).

Finally, for on-path truthfulness, we only need to write incentive constraints whenever the public history matches the private history. However, we do also need to update the definition of the intermediate variable $U_i(d+1; h^{[d-1]},\mathrm{\bar{h}}_i^{[d-1]}, v^{(d)}_i, v^{[d]:[D]}_{-i}, t)$ (the expected utility from stages $d+1$ onward when the public history is $h^{[d-1]}$, the private history is $\mathrm{\bar{h}}_i^{[d-1]}$ and the stage $d$ report is $v^{(d)}_i$ for buyer $i$ and the true stage $d$ type is $t$, while other buyer's reports are according to $v^{[d]:[D]}_{-i}$) to take into account the fact that an honest report in stage $d+1$ might not maximize the stage utility after a lie in stage $d$. Instead, we should consider every possible misreport $\hat{v}^{(d+1)}_i$:
\begin{multline*}
U_i(d+1; h^{[d-1]},\mathrm{\bar{h}}_i^{[d-1]}, v^{(d)}_i, v^{[d]:[D]}_{-i},t) \geq \sum_{v^{(d+1)}_i} Pr[v^{(d+1)}_i | \mathrm{\bar{h}}_i^{[d-1]},t] \cdot \\ 
\left( v^{(d+1)}_i x_i^{d+1}( [h^{[d-1]}, v^{(d)}_i, v^{(d)}_{-i}]; [\hat{v}^{(d+1)}_i, v^{(d+1)}_{-i}] ) - p_i^{d+1}([h^{[d-1]}, v^{(d)}_i, v^{(d)}_{-i}]; [\hat{v}^{(d+1)}_i, v^{(d+1)}_{-i}] )  \right. \\ \left. + U_i(d+2; [h^{[d-1]}, v^{(d)}_i, v^{(d)}_{-i}],[\mathrm{\bar{h}}_i^{[d-1]},t], \hat{v}^{(d+1)}_i, v^{[d+1]:[D]}_{-i}, v^{(d+1)}_i) \right).
\end{multline*}

The proof of correctness remains virtually unchanged.

%% file: different_auctions.tex

In this section we focus again on the single buyer setting.
So far, we have considered deterministic and randomized mechanisms. In this section we compare them in terms of the expected revenue generated, against each other and against two other benchmarks:

\begin{itemize}
\item the optimal non-adaptive mechanism --- i.e. running an independent Myerson's mechanism on each stage; and
\item the optimal social welfare SW --- the expected utility of the buyer from receiving both items for free.
\end{itemize}

The following is immediate:

\begin{fact}
For any distribution of valuations,
\[
\mbox{Rev}\left(\mbox{non-adaptive}\right)\leq\mbox{Rev}\left(\mbox{deterministic}\right)\leq\mbox{Rev}\left(\mbox{randomized}\right)\leq\mbox{SW}
\]

\end{fact}

But are these inequalities strict for some valuation distributions? And by how much? 


\begin{thm}\label{thm:separations}
Let $v^{*} = \max_{v \in V^{(1)} \cup V^{(2)} } v$ be the maximal buyer's valuation in any stage, and assume that all valuations are integral. Then in any two-stage mechanism, the maximum, over all mechanisms, ratio:

\begin{itemize}
\item between $\mbox{SW}$ and any of $\left\{ \mbox{Rev}\left(\mbox{non-adaptive}\right),\mbox{Rev}\left(\mbox{deterministic}\right),\mbox{Rev}\left(\mbox{randomized}\right)\right\} $
is exactly the harmonic number of $v^{*}$, $H_{v^{*}}=\sum_{i=1}^{v^{*}}1/i$;
\item between either of $\left\{ \mbox{Rev}\left(\mbox{deterministic}\right),\mbox{Rev}\left(\mbox{randomized}\right)\right\} $ and $\mbox{Rev}\left(\mbox{non-adaptive}\right)$ is at least $\Omega\left(\log^{1/2}v^{*}\right)$ (and at most $O\left(\log v^{*}\right)$); and
\item between $\mbox{Rev}\left(\mbox{randomized}\right)$ and $\mbox{Rev}\left(\mbox{deterministic}\right)$ is at least $\Omega\left(\log^{1/3}v^{*}\right)$ ( and at most $O\left(\log v^{*}\right)$ ).
\end{itemize}

Furthermore, even when the valuations on the different stages are independent, there exists a two-stage mechanism with ratio of $\Omega\left(\log\log v^{*}\right)$ between either of $\mbox{Rev}\left(\mbox{deterministic}\right)$, $\mbox{Rev}\left(\mbox{randomized}\right)$ and $\mbox{Rev}\left(\mbox{non-adaptive}\right)$.

\end{thm}


\subsection{Warm up: Revenue vs Social Welfare}

To compare non-adaptive mechanisms to optimal social welfare, we can assume without loss of generality that the mechanism occurs in a single stage.

\begin{lem}
\label{lem:SW_is_higher}Let $v^{*}$ be the maximal buyer's valuation, and assume that all valuations are integral. For a single stage mechanism, the maximum ratio between $\mbox{SW}$ and $\mbox{Rev}\left(\mbox{non-adaptive}\right)$ is at least $\frac{ \log\left(v^{*}\right)}{2}$.
\end{lem}

\begin{proof}
Suppose that the buyer has valuation $2$ with probability $1/2$, $4$ with probability $1/4$, etc. until $2^{n}$ with probability $2^{-n}$ (and $0$ also with probability $2^{-n}$), i.e. a truncated equal revenue distribution. The expected social welfare is $\mbox{SW}=\sum_{i=1}^{n}2^{-i}\cdot2^{i}=n$. 
For any choice of price $2^{k}$ chosen by the non-adaptive mechanism, the expected revenue is 
\[ 
\mbox{Rev}\left(\mbox{non-adaptive}\right)=2^{k}\cdot\sum_{i=k}^{n}2^{-i}<2.
\]
The lemma follows by noticing that $v^* = 2^n$.
\end{proof}

The construction above is extremely useful in proving such lower bounds. In fact it is also used in our \NP-hardness result. The distribution used is approximately the well known equal-revenue distribution. To unify our notation we refer to it as $\mbox{\sc pow2}\left[1,n\right]$. In general:

\begin{definition}
We say that $v\sim c\cdot\mbox{\sc pow2}\left[a,b\right]$ if $v=c\cdot2^{a+i}$ with probability $2^{-i-1}$ for all $i\in\left[b-a\right]$, and $v=0$ with probability $2^{a-b-1}$. Note in particular that the
expectation is 
\[
\mathbf{E}\left[\mbox{\sc pow2}\left[a,b\right]\right]=2^{a-1}\left(b-a+1\right)\mbox{ .}
\]
\end{definition}

We conclude this introductory subsection by proving a tight version of the above proposition, namely

\begin{lem}
Let $v^{*}$ be the maximal buyer's valuation, and assume that all valuations are integral. The maximum, over all single stage mechanisms,
ratio between $\mbox{SW}$ and $\mbox{Rev}\left(\mbox{non-adaptive}\right)$ is exactly the harmonic number of $v^{*}$.
\end{lem}

\begin{proof}
\begin{eqnarray*}
\mbox{SW} & = & \sum_{t=1}^{v^{*}}t\Pr\left[v=t\right]=\sum_{t=1}^{v^{*}}t\left(\Pr\left[ v \geq t\right]-\Pr\left[v\geq t+1\right]\right)=\sum_{t=1}^{v^{*}}\Pr\left[v \geq t\right]\\
 & = & \sum_{t=1}^{v^{*}} \frac{\mbox{Rev}\left(p=t\right)}{t} \leq\sum_{t=1}^{v^{*}} \frac{ \mbox{Rev}\left(\mbox{non-adaptive}\right)}{t} =\mbox{Rev}\left(\mbox{non-adaptive}\right)\cdot H_{v^{*}}
\end{eqnarray*}
where $\mbox{Rev}\left(p=t\right)$ denotes the expected revenue from charging $t$. Finally, note that the inequality can be made tight by setting $\Pr\left[v\leq t\right]=\frac{1}{t}$ for all $1\leq t\leq v^{*}$.
\end{proof}

In a single stage setting, the optimal randomized mechanism does not achieve more revenue than a posted price; therefore the same bound immediately holds for adaptive deterministic and randomized mechanisms.

\begin{cor}
Let $v^{*}$ be the maximal buyer's valuation, and assume that valuations are integral. The maximum, over all single stage mechanisms, ratio between $\mbox{SW}$ and 
any of $\mbox{Rev}\left(\mbox{non-adaptive}\right)$, $\mbox{Rev}\left(\mbox{deterministic}\right)$, or $\mbox{Rev}\left(\mbox{randomized}\right)$, is exactly the harmonic number of $v^{*}$.
\end{cor}

\subsection{Independent valuations}
Surprisingly, as we saw in the introduction, adaptive mechanisms achieve a higher revenue even when the valuations on the different stages are independent. Here, we show that large gaps exist even between deterministic adaptive mechanisms and non-adaptive mechanisms.

\begin{lem}
Let $v^{*}$ be the maximal buyer's valuation and assume that valuations are integral. For a two-stage mechanism, the ratio between $\mbox{Rev}\left(\mbox{deterministic}\right)$ and $\mbox{Rev}\left(\mbox{non-adaptive}\right)$ can be as large as $\frac{\left(\log\log v^{*}\right)}{4}$, \emph{even when the valuations on each stage are independent}.
\end{lem}

\begin{proof}
Let $N=2^{n}$. Let the valuation the first stage be distributed as $v^{(1)}\sim\mbox{\sc pow2}\left[1,n\right]$, and on the second stage $v^{(2)}\sim\mbox{\sc pow2}\left[1,N\right]$. As we have already seen in the introduction, the optimal revenue for running two separate fixed-price mechanisms is a small constant. Specifically, $\mbox{Rev}\left(\mbox{non-adaptive}\right)<4$.

What about deterministic adaptive mechanisms? The same idea works, except that in the deterministic case, the seller punishes the
buyer for lower bids by charging higher prices on the second stage.

On the first stage, the deterministic adaptive mechanism will charge the buyer almost her full reported value $v^{(1)}-\left(2-2^{-v^{(1)}}\right)$. If the first stage report is $v^{(1)}$, on the second stage, we offer the item for a price of $p^{(2)}\left(v^{(1)}\right)=2^{N-v^{(1)}}$. The buyer's expected utility from the second stage is now exactly 

\begin{align*}
\sum_{i\colon2^{i}\geq p^{(2)}\left(v^{(1)}\right)}2^{-i}\left(2^{i}-p^{(2)}\left(v^{(1)}\right)\right) 
	& =  v^{(1)}-\sum_{i\colon2^{i}\geq p^{(2)}\left(v^{(1)}\right)}2^{-i}p^{(2)}\left(v^{(1)}\right) \\
	& =v^{(1)}-\sum_{0\leq i\leq v^{(1)}-1}2^{-i} \\
	& =v^{(1)}-\left(2-2^{-v^{(1)}}\right)
\end{align*}

Once again, the buyer's expected utility on the second stage exactly covers the price on the first stage, which guarantees that this mechanism
satisfies IC.

Finally, note the expected revenue is almost as large as the expected valuation on the first stage $\mbox{Rev}\left(\mbox{deterministic}\right)>n-2$.
\end{proof}

\subsection{Stronger adaptivity gaps for correlated valuations}

When the valuations are correlated, we can show stronger adaptivity gaps.

\begin{lem}
\label{lem:strong_adaptivity_gap}
Let $v^{*}$ be the maximal buyer's valuation and assume that valuations are integral. For a two-stage mechanism, the ratio between $\mbox{Rev}\left(\mbox{deterministic}\right)$ and $\mbox{Rev}\left(\mbox{non-adaptive}\right)$ can be as large as $\sqrt{\log v^{*}}/4$
\end{lem}

\begin{proof}
Again, let the first-stage valuation be distributed $v^{(1)}\sim\mbox{\sc pow2}\left[1,n\right]$. The second-stage valuation $v^{(2)}$ will be conditioned on the first stage: $v^{(2)}\mid v^{(1)}\sim\left(v^{(1)}/n\right)\cdot\mbox{\sc pow2}\left[1,n^{2}\right]$. 

We already saw that the non-adaptive policy's revenue on the first
stage is less than $2$. What is the optimal price for the second stage?
To answer this question we must consider the marginal distribution
of the second stage:
\begin{gather*}
\Pr\left[v^{(2)}=2^{l}/n\right]=\sum_{k\in\left[n\right]}\Pr\left[v^{(1)}=2^{k}/n\right]\Pr\left[v^{(2)}=2^{l}\mid v^{(1)}=2^{k}\right]\leq\sum_{k\in\left[n\right]}2^{-k}2^{k-l}=n\cdot2^{-l}
\end{gather*}
Therefore, $\Pr\left[v^{(2)}\geq2^{l}/n\right]\leq n\cdot2^{1-l}$,
which implies $\mbox{Rev}\left(\mbox{non-adaptive}\right)<4$.

Now, consider the randomized mechanism that on the first stage charges
the buyer $v^{(1)}=2^{k}$ (and allocates the item), and on the second
stage allocates the item for free with probability $k/n$. When the
buyer's true valuation on the first stage is $2^{k}$, her the expected
utility from reporting $2^{l}$ is given by
\[
\mbox{U}\left(2^{k},2^{l}\right)=\left(l/n\right)\mathbf{E}\left[v^{(2)}\mid v^{(1)}=2^{k}\right]-2^{l}=l\cdot2^{k}-2^{l}\mbox{ ,}
\]
which is maximized by $l\in\left\{ k,k+1\right\} $. The expected
revenue from this randomized mechanism is again $n$.

Similarly, a deterministic mechanism can charge $v^{(1)}=2^{k}$ on the
first stage, and offer the item on the second stage for price $p^{(2)}\left(2^{k}\right)=2^{n^{2}-nk}/n$
\begin{eqnarray*}
\mbox{U}\left(2^{k},2^{l}\right) & = & \sum_{i\colon2^{k+i}/n\geq p^{(2)}\left(2^{l}\right)}2^{-i}\left(2^{k+i}/n-p^{(2)}\left(2^{l}\right)\right)-2^{l}\\
 & = & \left(l-\frac{k}{n}\right)2^{k}-\sum_{i\colon2^{k+i}/n\geq p^{(2)}\left(2^{l}\right)}2^{-i}\cdot p^{(2)}\left(2^{l}\right)-2^{l}\\
 & = & \left(l-\frac{k}{n}\right)2^{k}-\sum_{0\leq i\leq nl+k-1}2^{-i}\left(2^{k}/n\right)-2^{l}\\
 & = & \left(l-\frac{k}{n}\right)2^{k}-\left(2-2^{-\left(nl+k\right)}\right)\left(2^{k}/n\right)-2^{l}\\
 & = & \left(l+\frac{2^{-\left(nl+k\right)}}{n}\right)2^{k}-2^{l}-\left(\frac{k+2}{n}\cdot2^{k}\right)
\end{eqnarray*}
The second line follows because there are $nl-k$ $i$'s for which
$i\colon2^{k+i}/n\geq p^{(2)}\left(2^{l}\right)$. Notice that indeed,
$\left(l+\frac{2^{-\left(nl+k\right)}}{n}\right)2^{k}-2^{l}$ is maximized
at $l=k$.
\end{proof}

\subsection{Deterministic vs randomized mechanisms}

Naturally, one would expect that deterministic and randomized mechanisms yield different revenues because we can optimize the latter in polynomial time, while optimizing over deterministic mechanisms is NP-hard. In this subsection we show that randomized mechanisms can in fact yield much more revenue.

\begin{lem}
Let $v^{*}$ be the maximal buyer's valuation, and assume that all valuations are integral. For a two-stage mechanism, the ratio between $\mbox{Rev}\left(\mbox{randomized}\right)$ and $\mbox{Rev}\left(\mbox{deterministic}\right)$ can be as large as $\frac{\left(\log v^{*}\right)^{1/3}}{7}$.
\end{lem}

Our proof builds on the constructions in the proof of Lemma \ref{lem:strong_adaptivity_gap}. A key observation is that by modifying the parameters for the second stage distribution, we can shift the prices without changing the expected utility. Choosing those parameters based on the valuation in the first stage, will allow us to break the deterministic seller's strategy, without changing the revenue of the randomized mechanism.

\begin{proof}
Let $v^{(1)}\sim\mbox{\sc pow2}\left[1,n\right]$. For type $i$ with value $2^{i}$ on the first stage, the valuation on the second stage will be $0$ with probability $1-2^{-2n^{2}i}$. The remaining $2^{-2n^{2}i}$ will be distributed according to $\frac{2^{\left(2n^{2}+1\right)i}}{n}\mbox{\sc pow2}\left[1,n^{2}\right]$. For any $i\in\left[n\right]$, let $V_{i}^{(2)}\setminus\left\{ 0\right\} $ be the set of nonzero feasible valuations on the second stage, conditioned on valuation $2^{i}$ on the first stage. Notice that for any $i<j$, all the values in $V_{i}^{(2)}\setminus\left\{ 0\right\} $ are much smaller than all the values in $V_{j}^{(2)}\setminus\left\{ 0\right\} $.

The randomized mechanism, again charges full price $v^{(1)}=2^{k}$ on the first stage, and gives the item for free on the second stage, with probability $k/n$. The buyer's utility from reporting $2^{l}$ is:
\[
\mbox{U}\left(2^{k},2^{l}\right)=\left(l/n\right)\mathbf{E}\left[v^{(2)}\mid2^{k}\right]-2^{l}=l\cdot2^{k}-2^{l}\mbox{ ,}
\]
which is maximized by $l\in\left\{ k,k+1\right\} $. The expected revenue from this randomized mechanism is again $n$.

What about the deterministic seller? Given any deterministic mechanism, let $k^{*}$ be the minimal $k$ for which a buyer with first-stage valuation $2^{k}$ has a nonzero probability of affording both items. In other words, after declaring valuation $2^{k^{*}}$ for the first stage, her second-stage price is at most $p^{(2)}\left(2^{k^{*}}\right) \leq \frac{2^{\left(2n^{2}+1\right)k^*+1}}{n} < 2^{2n^{2}\left(k^*+1/2\right)}$.

Assume that the buyer has valuation $v^{(1)}=2^{l}>2^{k^{*}}$. If she deviates and declare type $2^{k^{*}}$, she receives the first item, and she also receives the second time whenever she has nonzero valuation. 
On the second stage, she pays less than  $2^{2n^{2}\left(k^*+1/2\right)}$ 
with probability $2^{-2n^{2}l} \leq 2^{-2n^{2}\left(k^*+1\right)}$.
Therefore her expected pay on the second stage has a negligible expected cost (less than $2^{-n^{2}}$). 
On the first stage, her price cannot be greater than $2^{k^{*}}$. The total expected payment made
by the buyer with $v^{(1)}=2^{l}>2^{k^{*}}$ is bounded by $2^{k^{*}}+2^{-n^{2}}$. 
Summing over the probabilities of having first-stage valuation $v^{(1)}=2^{l}>2^{k^{*}}$, this is still less than $1$.

Consider all the types whose first-stage valuations are lower than $2^{k^{*}}$, and yet they receive the first item. Since they can never afford the second item, on the first stage they must all be charged the same price, thus yielding a total revenue less than $2$. Similarly, the types for which the first-stage item is not allocated, must all be charged the same price on the second stage. Finally, by IR constraints the expected revenue from $v^{(1)}=2^{k^{*}}$is at most $2.$ Therefore, the total expected revenue is less than $7$.\end{proof}

%% file: no_contract.tex

\section{No contract\label{sec:No-contract}}

In this section, we restrict the two-stage mechanism problem to a setting where the seller cannot commit to an arbitrary action in the future (however the seller {\em can} commit to the mechanism within each stage).
There are indeed many well-studied situations in economics in which contracts about future transactions are impossible or legally problematic, especially when the time between stages is relatively long (e.g.~on the order of years). Beyond this consideration, in the two-stage no-contract case we study, the seller's optimal second stage mechanism is trivial (it's simply Myerson's optimal mechanism for the seller's updated prior), which raises hopes of escaping the negative results in the rest of this paper.

Before we proceed, let us specify the problem we study. In the {\sc No-contract two-stage mechanism} problem there is a single seller and a single buyer that interact over two stages. The seller has two items to sell, one in each stage. The buyer privately learns her first stage type $t_1 \in \Theta_1$, $t_1 = (v, D^{(2)})$, where $v$ is her value for the first stage item and $D^{(2)}$ is the distribution from which she'll draw her second stage item. The type $t_1$ is drawn from a publicly known distribution $D^{(1)}$ and occurs with probability $Pr[t_1]$. In the second stage, a buyer with first-stage type $t_1=(v, D^{(2)})$ draws her type/value for the second stage item, $t_2 \in \Theta_2$, from the distribution $D^{(2)}$.

Informally, the seller and buyer will interact during the first stage, the buyer will be allocated the first stage item with some probability for some payment, and the seller's prior over the second stage distribution will be updated based on this interaction; let $\mathcal{D}_2$ be the updated prior for the second stage. In the second stage, the item will be sold optimally with respect to the updated prior. The tension between the buyer and seller, as opposed to the problem with commitment, arises from the fact that the buyer might want to report her full type more than the seller wants to learn it. That is, reporting her full stage one type has a big impact on the second stage reserve price, but a small impact on the seller's second stage revenue. 

The seller will set up a communication protocol for the first stage, with a corresponding allowed set of messages that the buyer can send her, as well as functions/mappings from first stage types to messages. The first stage item will be allocated (and payments will be charged) according to the rules of this first stage protocol. In the second stage the seller will look at the first stage transcript (i.e., the messages exchanged), and update her prior over the buyer's second stage type \emph{assuming} the buyer used the prescribed mapping (from first stage types to first stage messages) in the past. Doing this Bayesian update is the only commitment the seller can make. Once this update is done, the seller will run an optimal mechanism for the second stage (optimal with respect to the updated prior). For the first stage protocol we still ask for dynamic incentive compatibility and ex-post individual rationality, formally defined later in this section.

Back to the formal description of the problem, a $2k-1$ round dynamic mechanism involves, for the first stage, (1) $k$ sets of messages $M_1, M_2, \dots, M_k$ that the buyer can send to the seller, (2) $k$ functions $f_1, \dots, f_k$, where $f_i$'s domain is $\Theta_1$, the set of first stage types, and its range is $M_i$, the $i$-th set of messages, (3) $k-1$ ``stopping'' functions $g_1, g_2, \dots, g_{k-1}$ from the set of messages received so far ($\times_{i=1}^{j} M_i$ for $g_j$) to the interval $[0,1]$ , (4) an allocation function $x_1$ and a payment function $p_1$ from the set of messages sent $\times_{i=1}^j M_i$ \footnote{For simplicity, we treat $j$, the number of messages sent by the buyer, as a constant, even though it is random. Equivalently, one can define the domain of $x_1$ and $p_1$ to be $\times_{i=1}^k M_i$, and query the functions only on $m_1, \dots, m_j, \bot, \dots, \bot$ when the mechanism is executed, where $\bot$ is an empty message.} to an allocation in $[0,1]$ and a payment in $\mathbb{R}$ for the first stage item, respectively.
In the second stage the seller re-inspects the messages $m_1, \dots, m_j$ she received in the first stage (where $1 \leq j \leq k$) and updates her prior over second stages types. Specifically, the seller assumes that the buyer followed the specified protocol and used functions $f_1, \dots, f_j$ to send messages. Therefore, the set of possible first stage types is $C = \cap_{i=1}^j f_i^{-1}(m_i)$, where we write $f^{-1}(x)$ for the preimage of a function, i.e. the set of values $A$, such that for all $y\in A$, $f(y) = x$. The updated second stage prior $\mathcal{D}_2$ is the distribution that first samples a first stage type $t_1 = (v, D)$ from $D^{(1)}$ (the first stage distribution) conditioned on $t_1 \in C$, and then samples the second stage value from $D$. If $C$ is empty, which implies that the buyer definitely did not follow the intended protocol, the seller uses the default prior $\mathcal{D}_2$, her prior before her interaction with the buyer.\footnote{We note that our results do not use this specific seller response in a crucial way.} The seller then runs the optimal mechanism for $\mathcal{D}_2$, which is just Myerson's mechanism. Therefore, for any given $k$, the seller's task is to design sets of messages $M_1, \dots, M_k$, reporting functions $f_1, \dots, f_k$, stopping functions $g_1, \dots, g_{k-1}$, and first stage allocation and payment functions $x_1$ and $p_1$; the mechanism in the second stage is the best one possible, assuming the prior update we described.

The order of events in a $2k-1$ round mechanism, for $k \geq 1$, is as follows:
\begin{enumerate}[label=(\roman*)]
\item The buyer privately learns her type $t_1 \in \Theta_1$, $t_1 = (v, D^{(2)})$. Set $j=1$.
\item The buyer sends a message $m_j \in M_j$ to the sender.
\item If $j =k$ we go to step (iv). If $j < k$ the seller flips a biased coin that lands heads with probability $g_j(m_1, \dots, m_j)$. If the coin lands heads the seller asks for the buyer to send another message: $j$ gets increased by $1$ and we go back to step (ii); otherwise, if the coin lands tails, we go to step (iv).
\item The first item is allocated with probability $x_1(m_1,\dots,m_j)$ and the buyer is charged  $p_1(m_1,\dots,m_j)$, where $j$ is the number of messages the seller has received from the previous interaction. The buyer gets expected stage utility $v \cdot x_1(m_1,\dots,m_j) - p_1(m_1,\dots,m_j)$, and we proceed to the second stage.
\item The buyer privately learns her type  (value for the second stage item) $t_2$, that is drawn from distribution $D^{(2)}$.
\item The seller updates her prior according to $m_1, \dots, m_j$ and $f_1, \dots, f_j$. Specifically, let $C = \cap_{i=1}^j f^{-1}_i(m_i)$ be the set of candidate first stage types. If $C$ is non-empty, let $\mathcal{D}_2$ be the distribution that first samples a type $t_1 = (v, D)$ according to $D^{(2)}$ conditioned on $t_1 \in C$, and then samples $v_2$ from $D$. If $C$ is empty, let $\mathcal{D}_2$ be the distribution that first samples a type $t_1 = (v,D)$ from $D^{(1)}$, and then samples $v_2$ from $D$. 
\item The seller runs the optimal one-shot mechanism for the distribution $\mathcal{D}_2$, i.e. sets a posted price equal to $p^* = argmax_p p \cdot Pr_{v_2 \sim \mathcal{D}_2}[ v_2 \geq p ]$. 
\item The buyer chooses to purchase the second stage item, if $t_2 \geq p^*$, and gets stage utility $t_2 - p^*$, or to not purchase and get stage utility zero.
\end{enumerate}

\begin{remark}[Randomized buyer's strategies] Although we model the buyer's strategy as deterministic, we note that our model can wlog accommodate buyer's randomized strategies by augmenting the buyer's stage 1 private type with a private tape of random bits. 
\end{remark}

\begin{remark}[Partial participation] Wlog the buyer has the option to report a special symbol and abort the first stage protocol without payment or allocation. In second stage the buyer can still participate in the mechanism, which will be the one that maximizes revenue for the seller conditioned on the buyer's type being such that they chose to abort in the first stage. 
\end{remark}

For the case of a single round mechanism, the seller simply designs a set of allowed messages and a mapping from types to messages (whose inverse will be used for the Bayesian update in the second stage). The buyer then sends a message, and the first stage item is sold according to this message. In isolation, this procedure is without loss of generality with respect to one-round protocols. That is, the seller can implement an arbitrary one-round communication protocol for selling the first item. This is not true for multiple rounds. There are more general multi-round mechanisms that the ones we consider here. For example, protocols where the buyer sends a message, and then the seller runs some protocol $X$ or some protocol $Y$ depending on this message, are outside our design space. Notice though that this only makes our result stronger, since in our main theorem for this section we are interested in upper bounding the revenue of single round mechanisms and lower bounding the revenue of multi-round mechanisms.

In order to formally define ex post IR and DIC, it will be useful to define some shorthand notation. Let $m^*_j = f_j( t_1 )$ be the message sent by a buyer with type $t_1$ who uses the function $f_j$. Let $\bar{m_j} = m_1 ,\dots, m_j$ be the first $j$ messages received by the seller during the execution of a mechanism.

\paragraph{Ex-post individual rationality} In the second stage every mechanism in this setting is, by definition, ex-post individually rational, no matter what has occurred in the first stage. For the first stage, we say that a $2k-1$ round dynamic mechanism is individually rational if for all first-stage types $t_1 \in \Theta_1$, for all stopping points $j \leq k$, the buyer that reports messages as prescribed by $f_1, \dots, f_k$ has non-negative stage utility (for the first stage) ex post. That is, for all $t_1 = (v, D^{(2)}) \in \Theta_1$, all stopping points $1 \leq j \leq k$, $v \cdot x_1( \bar{m^*_j} ) - p_1( \bar{m^*_j} ) \geq 0$. Note that even though this constraint is written in expectation, it is possible to guarantee non-negative utility ex-post by appropriately correlating the allocation and payment.

\paragraph{Incentive compatibility} Again, in the second stage every mechanism is incentive compatible by definition, since the buyer is faced with a posted price. For the first stage, we ask that sending messages according to $f_1, \dots, f_k$ maximizes the buyer's total utility, i.e. the sum of her stage utilities across stages. Let $D( \bar{m_j} )$ be the second stage prior the seller uses when she receives messages $\bar{m_j}$ in the first stage. Let $u_2 ( D' | D )$ be the expected second stage utility of a buyer whose value is drawn from a distribution $D$, but is faced with a posted-price calculated using a prior distribution $D'$.
Finally, let $Z(\bar{m_k})$ be the random variable for the stopping time of a $2k-1$ round mechanism, when the buyer sends messages $\bar{m_k} = m_1, \dots, m_k$. That is, $Z(\bar{m_k})$ takes the value $1$ with probability $1-g_1(m_1)$, the value $2$ with probability $(1-g_1(m_1))(1-g_2(m_1,m_2))$, and so on.

A $2k-1$ round dynamic mechanism satisfies dynamic incentive compatibility if, for all $t_1 \in \Theta_1$, $t_1 = (v, D^{(2)})$ and all messages $m_1 \in M_1, \dots, m_2 \in M_2, \dots, m_k \in M_k$, it holds that
\begin{multline*}
\mathbb{E}_{ j \sim Z(\bar{m^*_k}) } \left[  v \cdot x_1(\bar{m^*_j}) - p_1(\bar{m^*_j}) +  u_2 ( D(\bar{m^*_j}) | D ) \right]  \geq \\
\mathbb{E}_{ j \sim Z(\bar{m_k})} \left[  v \cdot x_1(\bar{m_j}) - p_1(\bar{m_j}) +  u_2 ( D(\bar{m_j}) | D ) \right].
\end{multline*}

\paragraph{Extensive communication increases revenue}

Our main result for this section is that there exists a $3$ round dynamic mechanism (i.e. uses two functions, $f_1$ and $f_2$) that makes more revenue than the optimal $3$ round dynamic mechanism.

\begin{restatable}[]{theom}{nocontract}
\label{thm:no-contract}The optimal mechanism for the {\sc No-contract two-stage mechanism} problem requires multiple rounds of communication in the
first stage.
\end{restatable}

We construct an instance where the buyer's valuations on the two stages are independent, yet on the first stage she has a more refined prior over her second-stage valuations. Furthermore, she has a strong incentive to share her information about the second period with the seller (while the seller is approximately indifferent). In order to {\em credibly} report her information about the second period, she needs the seller's help in setting up an incentive compatible mechanism. Informally, the seller now has another ``product'' she can sell for profit: the opportunity to report information about the second period. We will refer to this new product as {\em OTR}, for the ``opportunity to report.''

The OTR has two important properties that distinguish it from the real items sold in the mechanism: (1) because it is not a real item, it does not contribute to the buyer's valuation when evaluating the IR constraints; and (2) the seller knows its ex-interim value to the buyer. We'll set things up so that this value is independent of the partial information the buyer has in the first period; see the value of OTR in the statement of Lemma \ref{lem:D_1-and-D_2}. The latter property is useful when considering the DIC constraints.

How does the OTR lead to multiple rounds of communication? In order to satisfy the IR constraints, the OTR must be bundled with the first (real) item. Given the results from recent years about menu complexity (e.g \cite{hart2013menu,daskalakis2017strong}), it is unsurprising that the optimal way to sell this bundle is fractional; i.e. for each price $\pi$, the buyer receives the OTR with some probability $\rho^{\mbox{OTR}}\left(\pi\right)$ (and the real item with probability $1$ to satisfy the ex-post IR constraints). Thus we have: in round $1$, the buyer is asked to report her value (i.e. $f_1(v, D^{(2)}) = v$); in round $2$, the seller allocates the OTR with some probability that depends
on $v$; and in round $3$, if allocated the OTR, the buyer reports her information about the second stage (i.e. $f_2(v, D^{(2)}) = D^{(2)}$).

\input{no_contract_construction}

%% file: no_contract_construction.tex

\subsection{Construction}

For simplicity, we slightly alter our notation, and denote by $D_{1}$ and $D_{2}$ the possible distributions for the second stage valuation. 
Before the mechanism is executed, the seller has a prior of $\left(1/2,1/2\right)$ over $\left(D_{1},D_{2}\right)$, that is $Pr[ t_1 = (.,D_1) ] = Pr[ t_1 = (.,D_2) ] = 1/2$. We denote the mixed distribution known to the seller by $\left(\frac{1}{2}D_{1}+\frac{1}{2}D_{2}\right)$. We will introduce many constraints on these distributions, but the most important one for now is that the optimal posted price, which we'll refer to as the {\em Myerson price}, for each separate distribution is low (either $1$ or $1+\epsilon$ for $\epsilon\ll1$), while the Myerson price for the mixed distribution (i.e. the seller's prior) is high: $k$, for some sufficiently large integer $1\ll k\ll1/\epsilon$. 
In order to compensate a truthful buyer who may end up paying the slightly higher price ($1+\epsilon$) on the second stage, the seller gives her a discount of $\epsilon/5$ on the first stage. The following lemma lists all the properties we require from $D_{1}$ and $D_{2}$, as well as some useful notation. We will assume that we have such distributions for now, and construct them explicitly later.

\begin{restatable}[]{againlem}{nocontractdist}\label{lem:D_1-and-D_2}There
exist distributions $D_{1},D_{2}$ that satisfy all of the following
conditions:
\begin{description}
\item [{Myerson pricing}] The Myerson price given prior $D_{1}$ over the valuations is $1+\epsilon$, for prior $D_{2}$ it is $1$, and for prior $\frac{1}{2}D_{1}+\frac{1}{2}D_{2}$, it is $k$. Furthermore,
for any convex combination of $D_{1}$ and $D_{2}$, every Myerson price is one of these three possible prices.
\item [{Buyer's utility}] Let $u_{2}\left(D'\mid D\right)$ denote the buyer's expected utility from the second stage mechanism when her true distribution is $D$, but the seller runs a Myerson mechanism against a (possibly misreported) prior of $D'$. Then we require:

\begin{description}
\item [{Truthfulness\label{enu:truthfulness}}] $u_{2}\left(D_{1}\mid D_{1}\right)+\epsilon/5>u_{2}\left(D_{2}\mid D_{1}\right)$
and $u_{2}\left(D_{2}\mid D_{2}\right)>u_{2}\left(D_{1}\mid D_{2}\right)+\epsilon/5$.
(Note that in this case we can ensure incentive compatibility by giving
a discount of $\epsilon/5$ on the first stage whenever the buyer reports
$D_{1}$.)
\item [{Value of OTR\label{enu:Value-of-OTR}}]
The value of the OTR to the buyer does not depend on her private information. We use $\theta$ to denote this value. $\theta = u_{2}\left(D_{1}\mid D_{1}\right)+\epsilon/5-u_{2}\left(\frac{1}{2}D_{1}+\frac{1}{2}D_{2}\mid D_{1}\right)=u_{2}\left(D_{2}\mid D_{2}\right)-u_{2}\left(\frac{1}{2}D_{1}+\frac{1}{2}D_{2}\mid D_{2}\right)$. 
\end{description}
\item [{Seller's revenue\label{enu:revenue}}] By learning whether
the second stage's valuation is drawn from distribution $D_{1}$ or
$D_{2}$, the optimal expected revenue increases by at most $O\left(\epsilon\right)$.
\[
\textup{\texttt{Rev}}\left(\frac{1}{2}D_{1}+\frac{1}{2}D_{2}\right)\leq\frac{1}{2}\textup{\texttt{Rev}}\left(D_{1}\right)+\frac{1}{2}\textup{\texttt{Rev}}\left(D_{2}\right)\leq\textup{\texttt{Rev}}\left(\frac{1}{2}D_{1}+\frac{1}{2}D_{2}\right)+O\left(\epsilon\right)\mbox{.}
\]

\end{description}
\end{restatable}

The proof of Lemma~\ref{lem:D_1-and-D_2} can be found in~\ref{app:no_contract}.
Given the distributions guaranteed by Lemma \ref{lem:D_1-and-D_2} for the second stage,
we construct the following distribution over types for the first stage. It suffices to describe the distribution over values: the distribution over types will simply sample a value $v$, and then with probability $1/2$ the type will be $(v,D_1)$, otherwise $(v,D_2)$.
We want to construct a distribution over values whose optimal (posted price) revenue for the seller is significantly lower than the optimal social welfare (i.e. the buyer's expected valuation). Intuitively, the latter can only be translated into revenue by bundling with the OTR. In particular, this property can be achieved by the well-known equal-revenue distribution. Let $\delta$ be a small parameter ($\epsilon\ll\delta\ll1/k$), and let $l$ be a sufficiently large integer (say, $l=100$). We set

\[
\Pr\left[v^{(1)}=\delta j\right]=\begin{cases}
\frac{1}{j}-\frac{1}{j+1} & j\in\{1,\dots l-1\}\\
\frac{1}{j} & j=l
\end{cases}
\]

Let $R$ denote the expected revenue from the second stage when the seller knows which distribution is used, i.e. 
\[
R=\frac{1}{2}\textup{\texttt{Rev}}\left(D_{1}\right)+\frac{1}{2}\textup{\texttt{Rev}}\left(D_{2}\right).
\]
We prove that with three rounds of communication, the expected revenue is at least 
\begin{equation}
\textup{\texttt{Rev}}_{3}\geq R+\delta H_{l}-O\left(\epsilon\right),\label{eq:3-rounds}
\end{equation}
where $H_{l}$ is the $l$-th harmonic number. 

With one round, per contra, the expected revenue is at most 
\begin{equation}
\textup{\texttt{Rev}}_{1}\leq R+3\delta.\label{eq:1-round}
\end{equation}

\subsection{Three rounds of communication}

We begin by describing a protocol that uses multiple rounds of communication. $f_1(t_1)$ instructs the buyer to report just her value $v$ for the first stage item. Then, the seller flips a coin (whose bias depends on the value reported), and depending on the outcome, asks for the buyer to use $f_2(t_1)$, which instructs the buyer to report her second stage distribution. Thus, the intended order of events is:
\begin{enumerate}
\item The buyer reports her true valuation $v_{1}$;
\item With probability $v_{1}/\theta$, the seller allocates the OTR
(i.e. the seller asks the buyer for her prior on the second stage's valuation), where $\theta$ is defined in Lemma~\ref{lem:D_1-and-D_2};
\item If allocated the OTR, the buyer reports from which distribution ($D_{1}$
or $D_{2}$) she will draw her valuation on the second stage.
\end{enumerate}

Then, the first item is always allocated ($x_{1}=1$), and the seller charges price $p_{1}=v_{1}-\epsilon/5$ if the buyer reported prior
$D_{1}$, and $p_{1}=v_{1}$ otherwise. In the second stage, the second item is offered for the Myerson price for the seller's updated prior.

The ex-post IR constraints are trivially satisfied because $p_{1}\leq v_{1}$, and the first item is always allocated. The expected revenue from this mechanism, assuming the buyer reports truthfully, is as follows. On the first stage, the revenue is at least $\E\left[v_{1}\right]-\epsilon/5=\delta H_{l}-O\left(\epsilon\right)$. On the second stage, by ``Seller's revenue'' in Lemma \ref{lem:D_1-and-D_2}, the expected revenue is at least $R-O\left(\epsilon\right)$. Overall we match our guarantee (\ref{eq:3-rounds}).

Finally, we prove that this mechanism satisfies the DIC constraints. Given that she is allocated the OTR, it follows by ``Truthfulness''
in Lemma \ref{lem:D_1-and-D_2} that the buyer maximizes her utility by reporting truthfully in the third round. We now consider two cases, based on the buyer's prior for the second stage. We prove that in both cases, the buyer's expected utility does not depend on the reported valuation (and thus satisfies DIC):
\begin{itemize}
\item [{$D_{1}$:}] The buyer's expected utility on the second stage from the OTR is $\theta-\left(\epsilon/5\right)$. Upon reporting a first stage valuation $v'$, she receives the OTR with probability $v'/\theta$, so her added utility is $v'-\left(\epsilon/5\right)v'/\theta$. Similarly, on the first stage her expected price is $v'-\left(\epsilon/5\right)v'/\theta$. Her total utility is therefore independent of the valuation she reports in the first stage. 
\item [{$D_{2}$:}] Analogously to the previous case, the buyer's expected utility on the second stage from the OTR is $\theta$. Upon reporting
a first stage valuation $v'$, she receives the OTR with probability $v'/\theta$, so her added utility is $v'$. On the first stage her price is always $v'$. Her total utility is again independent of the valuation she reports in the first stage. 
\end{itemize}

\subsection{One round of communication}

Our goal in this subsection will be to upper bound the revenue extracted by any protocol with a single round of communication, when the buyer distributions are as previously described.
Consider any protocol with a single round of communication, that is the seller designs a single function $f = f_1$ from types to messages. 
Recall that by our construction for the second stage distributions, for any convex combination of $D_{1}$ and $D_{2}$, the Myerson price
charged in the second stage by the greedy seller is always one of three possibilities: $p_{2}\in\{1,1+\epsilon,k\}$. The choice of $p_{2}$ depends on the updated prior based on the buyer's single message. 

We divide the universe of legal buyer's messages into three: $M_{p}\subset\Sigma^{*}$ is the subset of messages for which the second stage price is $p$, for each of $p\in\{1,1+\epsilon,k\}$. Note that for each subset $M_{p}$, the allocation and price on the second stage are independent of the choice of message $m\in M_{p}$. In particular, any difference between messages must be due to different outcomes on the first stage.

Fix any $p\in\{1,1+\epsilon,k\}$, and let $R_{p}^{(1)}$ be the expected revenue on the first stage when the buyer's message is in $M_{p}$.
We show that $R_{p}^{(1)}\leq\delta$ by a reduction to selling only the first item. For message $m\in M_{p}$, let $x_{m}$ denote the probability that the seller allocates the first item to the buyer, and let $\pi_{m}$ denote the expected price when allocated. In fact,
since both the buyer and the seller are risk-neutral, we can assume without loss of generality that the seller charges exactly $\pi_{m}$ whenever the first item is allocated (note that ex-post IR constraints are preserved). Consider the single-item mechanism (for the first item) which requires the buyer to submit a message $m\in M_{p}$, and then with probability $x_{m}$ offers the item for price $\pi_{m}$. By the DIC constraints, whenever the buyer submitted a message $m\in M_{p}$ in the original (two-stage) mechanism, she will continue to submit the same message in the modified (single-item) mechanism. Therefore, the revenue collected from this single-item mechanism is at least $R_{p}^{(1)}$. Finally, observe that due to the equal-revenue construction, the revenue from selling the first item independently is at most $\delta$.

Adding the expected revenues from all the feasible $p$'s we have that the total expected revenue on the first stage at most $3\delta$. Since the revenue on the second stage is always at most $R$, (\ref{eq:1-round}) follows. This completes the proof of Theorem \ref{thm:no-contract}.\qed

%% file: discussion.tex

In this paper we studied dynamic mechanisms with full commitment and limited commitment. When the seller can commit to a contract, we showed that computing the optimal deterministic mechanism for the simple two-stage case is an NP-hard problem, and identified some tractable special cases: independent stages, fixed first stage prices and constant first stage support. The optimal randomized mechanism can be computed via a linear program whose size is polynomial in the support of the type distribution, even when we consider multiple stages and multiple buyers. We also proved that when the seller cannot commit to a future contract (but has limited commitment power) we have a very different kind of obstacle to overcome: multiple communication rounds.

There are still many interesting directions to pursue. In the two-stage deterministic case, constant approximations might still be achievable in the general, NP-hard case (our current proof does not even establish APX-completeness). Also, there may be other tractable special cases. Our reduction constructs complicated second-stage distributions. For example, what if the valuation distribution is ``affiliated'' (higher first stage valuation implies higher second stage valuation)? Is this case tractable?  And if so, can these results be extended to multiple stages and multiple buyers?

Another open question, from an algorithmic point of view, is relaxing the ex-post IR constraint. In the economics literature we have seen something similar in the work of~\cite{courty2000sequential} where airplane tickets refunds are sold before the agents see their valuation for them, and refunds are allowed. Is the two-stage mechanism equivalently tractable in this case?

Regarding the limited commitment case, an interesting open problem is whether there exists an examples with arbitrarily many rounds of communication. Also, can we quantify the revenue loss one occurs by restricting to a single round of communication? Can we find the optimal mechanism efficiently?

%% file: appendix.tex
\begin{appendix}

\section{\NP-hardness: Finalizing the construction}
\label{app:reduction_construct}

One of the most important parameters in our construction is $r_i$: 
we later prove that $r_i$ is the difference in expected revenue, conditioned on type $t_i$, between pricing at $(B_i, C_i)$, and pricing at $(A_i, D_i)$.

We set $r_{n+1} = \frac{(A_{n+1} - P^*)(\gamma-1)}{2\gamma^{4(n+1)}} = \Theta(n)$; the rest of $r_i$'s are defined recursively:
\begin{equation} \label{eq:recursion} 
r_i = \gamma^4 r_{i+1} - (\gamma - 1)[ \epsilon(\gamma^3 - \gamma) + \gamma ] .
\end{equation}
Notice that $\frac{r_1}{r_{n+1}} \leq \gamma^{4(n+1)} = \Theta(1)$.

Let $C_i = \frac{\frac{\gamma}{\gamma-1}r_i - \epsilon}{h_i}$ and $D_i = \frac{\gamma r_i}{ (\gamma - 1)h_i}$. Observe that with the recursive definition of $r_i$ (\ref{eq:recursion}) we can get a nice expression for the following difference:
\[
 C_{i+1}-D_i = \gamma^2 \frac{1-\epsilon}{h_i} .
\]
 The differences between pairs of special prices are summarized in Table~\ref{table:distances}.

Finally, we want the contribution towards the revenue from each vertex in the independent set to be the same. To that end, we set $r = \sum 1/r_i  = \Theta(1)$, and weight the probability of observing each type $t_i$ by $w_i = r/r_i$. We set the total probability of observing any of the $t_i$'s to be $p = \frac{\epsilon}{16nr} = \Theta(n^{-3})$.

\renewcommand{\arraystretch}{2}
\begin{center}
\begin{table}[h!]
\centering
\begin{tabular}{ | l | l | l | l | l | l | l | l}
\cline{1-7}
	\multirow{2}{*}{Type }
	& \multirow{2}{*}{$0 \rightarrow C_i$ }
	& \multirow{2}{*}{$C_i \rightarrow D_i$ }
	& \multirow{2}{*}{$D_i \rightarrow C_{i+1}$ }
	& \multirow{2}{*}{$D_{j-1} \rightarrow C_{j}$ }
	& \multirow{2}{*}{$C_j \rightarrow D_j$ }
	& $2^k D_{n+1} \rightarrow $\\ &&&&&& $2^{k+1} D_{n+1}$ \\ 
\cline{1-7} 
	\multirow{2}{*} {$ \bar{F}_i$} 
	& \multirow{2}{*}{$h_i$}  
	& \multirow{2}{*}{$\frac{h_i}{\gamma}$} 
	& $\frac{h_i}{\gamma^2}$ 
	& $\frac{1 - \epsilon (2 - \frac{1}{\gamma})}{1 - \epsilon} \frac{h_{j-1}}{\gamma^2}$ 
	& $(2-\frac{1}{\gamma})h_j $ 
	& \multirow{2}{*}{$\frac{h_{n+1}}{2^{k+1} \gamma}$}	
	& edge \\
\cline{4-6}
	& & 
	& $\frac{1 - \frac{\epsilon}{\gamma}}{1 - \epsilon}\frac{h_i}{\gamma^2} $  
	& $\frac{h_{j-1}}{\gamma^2}$ 
	& $h_j $ 
	& & no edge\\
\cline{1-7}
	{$\bar{F}_*$} 
	&  
	& $h_i$ 
	& $\frac{h_i}{\gamma^2}$ 
	& $\frac{h_{j-1}}{\gamma^2}$ 
	& $h_j$ 
	& $h_*$ \\
\cline{1-7}
\end{tabular}
\caption{\label{table:F_i} Cumulative distributions}
\end{table}
\end{center}

\renewcommand{\arraystretch}{2}
\begin{center}
\begin{table}[ht!]
\centering
\begin{tabular}{ | l | l | l | l | l | l | l | l}
\cline{1-6}
	   $B_i - A_i$
	& $A_i - B_{i+1}$
	& $A_{n+1} - P^*$
	& $D_i - C_i$
	& $C_{i+1} - D_i$
	& $Q^* - D_{n+1}$\\
	
\cline{1-6} 
	   $\epsilon$
	& $1-\epsilon$
	& $n^2 + n -\epsilon$
	&$\frac{\epsilon} {h_i}$
	& $\gamma^2 \cdot \frac{1-\epsilon}{h_i}$
	& $\left( 2^{8\gamma^{4(n+1)}} -1 \right) D_{n+1}$\\

\cline{1-6}
\end{tabular}
\caption{\label{table:distances} Differences between prices}
\end{table}
\end{center}

Recall that the IC constraints depend on the integrals of the cumulative distribution functions.
The values of the $\bar{F}_i$'s and $\bar{F}_*$ in our construction are tailored to make sure that their integrals have the values described in Table \ref{table:integrals}.

\newcommand{\tablebra}
{
\renewcommand{\arraystretch}{1.5}
\begin{center}
\begin{table}[!ht]
\centering
\begin{tabular}{ | l | l | l | l | l | l | l }
\cline{1-6}
	Type 
	& $C_i \rightarrow D_i$ 
	& $D_i \rightarrow C_{i+1}$ 
	& $D_{j-1} \rightarrow C_{j}$ 
	& $C_j \rightarrow D_j$ 
	& $D_{n+1} \rightarrow Q^*$\\ 
\cline{1-6}
	\multirow{3}{*}{$\int \bar{F}_i$} 
	& \multirow{2}{*}{$\frac{\epsilon}{\gamma}$} 
	& $1-\epsilon$ 
	& $1-(2-\frac{1}{\gamma})\epsilon$ 
	& $(2-\frac{1}{\gamma})\epsilon$ 
	& \multirow{3}{*}{$A_{n+1} - P^*$} 
	& edge \\ 
\cline{3-5}
	&  & $1-\frac{\epsilon}{\gamma}$  
	& $1-\epsilon$ 
	& $\epsilon$ 
	&  & no edge\\
\cline{2-5}
	& \multicolumn{4}{|c|}{$\int_{C_i}^{D_j} \bar{F}_i = B_i - A_j$} & \\
\cline{1-6}
	\multirow{2}{*}{$\int \bar{F}_*$} 
	& $\epsilon$ 
	& $1-\epsilon$ 
	& $1-\epsilon$ 
	& $\epsilon$ 
	& $A_{n+1} - P^*$ \\
\cline{2-6}
	& \multicolumn{5}{|c|}{$\int_{C_i}^{Q^*} \bar{F}_* = B_i - P^*$} & \\
\cline{1-6}
\end{tabular}
}
\newcommand{\tableket}
{
\end{table}
\end{center}
}

\tablebra
\caption{\label{table:integrals} Integrals of cumulative distributions.}
\tableket

\begin{claim}
\label{prop:integrals}
The integrals of the $\bar{F}_i$'s and $\bar{F}_*$ have values as stated in the Table~\ref{table:integrals}.
\end{claim}

\begin{proof}
Follows from multiplying the correct combination of entries of Table~\ref{table:F_i} and Table~\ref{table:distances}.
\end{proof}

\bigskip\noindent This completes the construction of the instance of {\sc Two-stage Mechanism}, starting from the instance of {\sc Independent Set}.  Notice that the numbers used are polynomial in the size of the input graph.

\section{\NP-hardness: Soundness}
\label{app:reduction_sound}

\begin{lem}
\label{lem:soundness}
Let $S$ be a maximum independent set in $G$. Then any IC and IR
mechanism has expected revenue at most 
\begin{equation}
\left(1-p\right) Rev\left(t^*,P^{*},Q^{*}\right)+p\sum_{i\in V} w_i Rev(t_i,B_{i},D_{i})+pr\left|S\right|. \label{eq:max-rev}
\end{equation}

\end{lem}

\subsection{Proof outline}

We first show that charging the pair $\left(P^{*},Q^{*}\right)$ maximizes the revenue that can be obtained from type $t^*$ (Claim \ref{cla:p^star-q^star_optimal}), and that $\left(B_{i},C_{i}\right)$ yields the optimal revenue from type $t_i$ (Claim \ref{cla:C_i-is-opt}). Observe that even if we could charge the optimal prices from every type, our expected revenue would be $(1-p)Rev(t^*,P^{*},Q^{*})+p\sum w_i Rev(t_i,B_i,C_i)$, which improves over (\ref{eq:max-rev}) by less than $prn=\epsilon/16$. Intuitively, this means that any deviation that results in a loss of $prn$ in terms of revenue, cannot compete with (\ref{eq:max-rev}).

Next, we show (Claim \ref{lem:violating_edge}) that if $\left(i,j\right)\in E$, then we cannot charge both $t_i$ and $t_j$ the optimal prices $\left(B_{i},C_{i}\right)$ and $\left(B_{j},C_{j}\right)$. In fact, we need a robust version
of this statement: Specifically, for some small parameters $\zeta^{(1)},\zeta_{i}^{(2)}$ (to be defined later), we show that we cannot charge both $t_i$ and $t_j$ prices in $\left[B_{i}-\zeta^{(1)},B_{i}\right]\times\left[C_{i}-\zeta_{i}^{(2)},C_{i}\right]$ and $\left[B_{j}-\zeta^{(1)},B_{j}\right]\times\left[C_{j}-\zeta_{j}^{(2)},C_{j}\right]$, respectively.

What can we charge type $t_i$ instead? In Claim \ref{lem:star-constraints-are-tight} we show that charging less than $C_{i}$ would require us to either not sell the item on the first stage, or charge type $t^*$ less than the optimal price. On the former case, we would lose $pw_{i}\cdot B_{i}>\epsilon/16$ revenue, and would immediately imply smaller revenue than (\ref{eq:max-rev}). On the latter case, we can use the robustness of Claim \ref{lem:violating_edge};
namely, we use the fact that we cannot charge $i$ prices that are $\left(\zeta^{(1)},\zeta_{i}^{(2)}\right)$-close to $\left(B_{i},C_{i}\right)$. This will imply that we must change the prices for type $t^*$ by some $\zeta_{*}^{(1)}$ on the first stage or $\zeta_{*}^{(2)}$ on the second stage. In either case the lost revenue is again greater than what we could potentially gain over (\ref{eq:max-rev}). Therefore, we must charge $t_i$ more than $C_{i}$ on the second stage. Claim \ref{lem:D_i-is-2nd} shows that charging $D_{i}$ is the best option in this case.

Therefore an upper bound to the revenue we can make is the following: charge $\left(B_{i},C_{i}\right)$ for all $i$ belonging to some independent set $S'$, and $\left(B_{j},D_{j}\right)$ for all other $j\notin S'$. (It is easy to see than in our construction even these prices won't satisfy the IC constraints). Now, the revenue given by these prices is:
\[
\left(1-p\right) Rev\left(t^*,P^{*},Q^{*}\right)+p\sum_{i\in S'}w_{i} Rev\left(t_i,B_{i},C_{i}\right)+p\sum_{j\notin S'}w_{j} Rev\left(t_j,B_{j},D_{j}\right).
\]
Notice that 
\begin{align*}
\sum_{i\in S'}w_{i}Rev\left(t_i,B_{i},C_{i}\right) &\leq \sum_{i\in S'}w_{i}\left(Rev\left(t_i,B_{i},C_{i}\right)-Rev\left(t_i,A_{i},D_{i}\right)+Rev\left(t_i,B_{i},D_{i}\right)\right)\\
 &=  \sum_{i\in S'}w_{i}\left(r_{i}+Rev\left(t_i,B_{i},D_{i}\right)\right).
\end{align*}

Therefore, the total expected revenue is
\[
\left(1-p\right)Rev\left(t^*,P^{*},Q^{*}\right)+pr|S'| + p\sum_{i\in S'}w_{i}Rev\left(t_i,B_{i},D_{i}\right)+p\sum_{j\notin S'}w_{j}Rev\left(t_j,B_{j},D_{j}\right),
\]

which is at most the expression in (\ref{eq:max-rev}).

\subsection{Preliminaries}

We begin by setting our padding parameters: let $\zeta^{(1)}=\frac{\epsilon}{4}$, and for each $i$ let $\zeta_{i}^{(2)}=\frac{\epsilon}{4\gamma^{2}h_{i}}$. In particular, this implies that for every $i$, $\zeta_{i}^{(2)}h_{i}+\zeta^{(1)}<\frac{\epsilon}{2}<\epsilon-\epsilon'$. Next, let $\zeta_{*}^{(1)}=\frac{\epsilon}{8}$, and $\zeta_{*}^{(2)}=\frac{\epsilon}{8h_*}$. We now have that $\zeta_{i}^{(2)}h_{i}\gamma^{2}=\zeta^{(1)}=\zeta_{*}^{(2)}h_{*}+\zeta_{*}^{(1)}$, which we will use later in the proof. Most importantly, recall that losing $\frac{\epsilon}{8}$ from the revenue from type $t^*$, is equivalent to a loss of $\left(1-p\right)\frac{\epsilon}{8} > \frac{\epsilon}{16}$ from the total expected revenue, which immediately implies that the expected revenue is less than (\ref{eq:max-rev}).

\subsection{Optimality of \texorpdfstring{$\left(P^{*},Q^{*}\right)$}{Lg}}

We now prove that prices $\left(P^{*},Q^{*}\right)$ maximize
the revenue from type $t^*$, in a robust sense:
\begin{claim}
\label{cla:p^star-q^star_optimal}Charging type $t^*$ prices $\left(P^{*},Q^{*}\right)$ maximizes the revenue from that type. Furthermore, if $p_{*}<P^{*}-\zeta_{*}^{(1)}$ or $q_{*}<Q^{*}-\zeta_{*}^{(2)}$, then the revenue from type $t^*$ is lower than the maximal revenue by at least $\zeta_{*}^{(1)}$ or $\zeta_{*}^{(2)}h_{*}$, respectively.\end{claim}

\begin{proof}
Clearly, $P^{*}$ is the most that we can charge type $t^*$ on the first stage. It is left to show that $Q^{*}$ maximizes the revenue on the second stage. 

On the second stage, we have:
\[
Rev^{(2)}\left(t^{*},Q^{*}\right)=Q^{*}h_{*}>A_{n+1}-P^{*}.
\]
Recall that $\bar{F}_{*}$ changes on $C_{i}$'s and $D_{i}$'s, so those are the only candidates we should compare with $Q^{*}$. For any $C_{i}$, we have 
\begin{align*}
Rev^{(2)}\left(t^{*},C_{i}\right) =  C_{i}h_{i}\gamma^{2}<\frac{\gamma^{3}r_{i}}{\gamma-1} \leq  \frac{\gamma^{3}r_{1}}{\gamma-1}\leq\frac{\gamma^{4\left(n+1\right)}r_{n+1}}{\gamma-1}=\frac{A_{n+1}-P^{*}}{2}.
\end{align*}

Similarly, for $D_{i}$,
\begin{align*}
Rev^{(2)}\left(t^{*},D_{i}\right) &= D_{i}h_{i}<\frac{\gamma r_{i}}{\gamma-1}<\frac{A_{n+1}-P^{*}}{2}. \qedhere
\end{align*}

\end{proof}

\subsection{Optimality of \texorpdfstring{$\left(B_{i},C_{i}\right)$}{Lg}}

Similarly, we show that $\left(B_{i},C_{i}\right)$ maximize the revenue from type $t_i$.

\begin{claim}
\label{cla:C_i-is-opt}$\forall x\neq C_{i}\,\,\,Rev^{(2)}\left(t_i,C_{i}\right)>Rev^{(2)}\left(t_i,x\right).$ \end{claim}
\begin{proof}
Since $\bar{F}_{i}$ is constant for all $x\leq C_{i}$, the claim for this
domain follows trivially. We will prove that $Rev^{(2)}\left(t_i,C_{i}\right)>Rev^{(2)}\left(t_i,D_{i}\right)$
and deduce from Claim \ref{lem:D_i-is-2nd} that the claim continues
to holds for any other $x$.

\[
Rev^{(2)}\left(t_i,C_{i}\right) =  C_{i}\cdot F_{i}\left(C_{i}\right)=\frac{\gamma}{\gamma-1}r_{i}-\epsilon = \frac{r_{i}}{\gamma-1}+r_{i}-\epsilon>\frac{r_{i}}{\gamma-1}=Rev^{(2)}\left(t_i,D_{i}\right). \qedhere
\]

\end{proof}

\subsection{Condition on edges}

Below we show that if there is an edge $(i,j)$, then we cannot charge both $t_i$ and $t_j$ close to their optimal prices:

\begin{claim}
\label{lem:violating_edge}If $\left(i,j\right)\in E$ then it cannot
be that $\left(p_{i},q_{i}\right)\in\left[B_{i}-\zeta^{(1)},B_{i}\right]\times\left[C_{i}-\zeta_{i}^{(2)},C_{i}\right]$
and $\left(p_{j},q_{j}\right)\in\left[B_{j}-\zeta^{(1)},B_{j}\right]\times\left[C_{j}-\zeta_{j}^{(2)},C_{j}\right]$.\end{claim}
\begin{proof}
Without loss of generality, let $i<j$. Assume (towards contradiction) that the conclusion is false. 
Then we get $\int_{q_{i}}^{q_{j}}\bar{F}_{i}<p_{i}-p_{j}$, which is a contradiction
to the IC constraints for type $t_i$: 
\begin{eqnarray*}
\int_{q_{i}}^{q_{j}}\bar{F}_{i} & = & \int_{q_{i}}^{C_{i}}\bar{F}_{i}+\int_{C_{i}}^{C_{j}}\bar{F}_{i}+\int_{C_{j}}^{q_{j}}\bar{F}_{i}\\
 & \leq & \int_{C_{i}}^{C_{j}}\bar{F}_{i}+\zeta_{i}^{(2)}h_{i}=j-i-\epsilon+\epsilon'+\zeta_{i}^{(2)}h_{i}\\
 & < & j-i-\zeta^{(1)}\\
 & = & B_{i}-B_{j}-\zeta^{(1)}\\
 & \leq & p_{i}-p_{j}\mbox{ ,}
\end{eqnarray*}
where the third line follows by $\zeta_{i}^{(2)}h_{i}+\zeta^{(1)}<\epsilon-\epsilon'$. 
\end{proof}

\subsection{Restriction imposed by charging \texorpdfstring{$\left(P^{*},Q^{*}\right)$}{Lg} for type
\texorpdfstring{$t^*$}{Lg}}

The claim below essentially shows that we cannot go around the restriction on prices for neighbors by reducing the prices:

\begin{claim}
\label{lem:star-constraints-are-tight}If $p_{*}>P^{*}-\zeta_{*}^{(1)}$
and $q_{*}>Q^{*}-\zeta_{*}^{(2)}$, then in any IC solution either:
\begin{itemize}
\item $p_{i}>B_{i}$ - note that this means that type $t_i$ cannot purchase the item on the first stage; or
\item $q_{i}>C_{i}$ - note that this substantially decreases our revenue for type $t_i$ on the second stage; or
\item $p_{i}\geq B_{i}-\zeta^{(1)}$ and $q_{i}\geq C_{i}-\zeta_{i}^{(2)}$.
\end{itemize}
\end{claim}

\begin{proof}

The negation of the claim gives us two configurations: having $p_{i}\leq B_{i}$ and $q_{i}<C_{i}-\zeta_{i}^{(2)}$, and having $p_{i}<B_{i}-\zeta^{(1)}$ and $q_{i}\leq C_{i}$. We will show the claim is true by contradiction, i.e. both these configurations are violating.

Assume first that $p_{i}\leq B_{i}$ and $q_{i}<C_{i}-\zeta_{i}^{(2)}$. Consider the IC constraint comparing $t^*$'s utility when telling the truth and when claiming that she is type $t_i$:
\begin{align*}
\int_{q_{i}}^{q_{*}}\bar{F}_{*} & =\int_{q_{i}}^{C_{i}}\bar{F}_{*}+\int_{C_{i}}^{Q^{*}}\bar{F}_{*}+\int_{Q_{*}}^{q^{*}}\bar{F}_{*}\\
 & >\int_{C_{i}-\zeta_{i}^{(2)}}^{C_{i}}\bar{F}_{*}+\int_{C_{i}}^{Q^{*}}\bar{F}_{*}+\int_{Q_{*}}^{Q^{*}-\zeta_{*}^{(2)}}\bar{F}_{*} \\
 &=\int_{C_{i}}^{Q^{*}}\bar{F}_{*}+\zeta_{i}^{(2)}\frac{h_{i-1}}{\gamma^{2}}-\zeta_{*}^{(2)}h_{*}\\
 & =\int_{C_{i}}^{Q^{*}}\bar{F}_{*}+\zeta_{*}^{(1)}=B_{i}-P^{*}+\zeta_{*}^{(1)}\\
 & \geq p_{i}-p_{*},
\end{align*}
where the third line follows from $\zeta_{i}^{(2)}\frac{h_{i-1}}{\gamma^{2}}=\zeta_{*}^{(2)}h_{*}+\zeta_{*}^{(1)}$.

We now return to the other violating configuration, namely $p_{i}<B_{i}-\zeta^{(1)}$
and $q_{i}\leq C_{i}$. We now have
\begin{align*}
\int_{q_{i}}^{q_{*}}\bar{F}_{*} &= \int_{q_{i}}^{C_{i}}\bar{F}_{*}+\int_{C_{i}}^{Q^{*}}\bar{F}_{*}+\int_{Q_{*}}^{q^{*}}\bar{F}_{*}\\
 &>  \int_{C_{i}}^{C_{i}}\bar{F}_{*}+\int_{C_{i}}^{Q^{*}}\bar{F}_{*}+\int_{Q_{*}}^{Q^{*}-\zeta_{*}^{(2)}}\bar{F}_{*} \\
 &=\int_{C_{i}}^{Q^{*}}\bar{F}_{*}-\zeta_{*}^{(2)}h_{*}\\
 &= \int_{C_{i}}^{Q^{*}}\bar{F}_{*}-\zeta^{(1)}+\zeta_{*}^{(1)}=B_{i}-\zeta^{(1)}-P^{*}+\zeta_{*}^{(1)}\\
 &\geq p_{i}-p_{*},
\end{align*}
where the third line follows from $\zeta_{i}^{(2)}\frac{h_{i-1}}{\gamma^{2}}=\zeta_{*}^{(2)}h_{*}+\zeta_{*}^{(1)}$.
\end{proof}

\subsection{Optimality of \texorpdfstring{$\left(B_{i},D_{i}\right)$}{Lg}}

We now show that $D_{i}$ is the optimal price on the second stage for type $t_i$, conditioned on charging more than $C_{i}$.

\begin{claim}
\label{lem:D_i-is-2nd}$\forall y>C_{i}\,\,\,Rev^{(2)}\left(t_i,D_{i}\right)\geq Rev^{(2)}\left(t_i,y\right)$ \end{claim}
\begin{proof}
It is easy to see that the second stage revenue is maximal for one of
the ``special points'' where $F_{i}$ changes. At $D_{i}$ we have:
\begin{gather*}
Rev^{(2)}\left(t_i,D_{i}\right)=D_{i}\cdot \bar{F}_{i}\left(D_{i}\right)=\frac{\gamma r_{i}}{\left(\gamma-1\right)h_{i}}\cdot\frac{h_{i}}{\gamma_{i}}=\frac{r_{i}}{\gamma-1}.
\end{gather*}
We now compare with each of type of special point:
\begin{itemize}[leftmargin=*]
\item What happens if we set $q_{i}=C_{i+1}$?
\begin{align*}
Rev^{(2)}\left(t_i,C_{i+1}\right) &= C_{i+1}\cdot \bar{F}_{i}\left(C_{i+1}\right) \\
&\leq\frac{\frac{\gamma}{\gamma-1}r_{i+1}-\epsilon}{h_{i}\gamma^{-4}}\cdot\frac{h_{i}}{\gamma^{2}}\left(\frac{1-\epsilon/\gamma}{1-\epsilon}\right)\\
&\leq  \frac{\gamma^{5}r_{i+1}}{\gamma^{2}\left(\gamma-1\right)}\left(1+2\epsilon\right)\\
&=\frac{\gamma r_{i}+\left(\gamma-1\right)\left[\epsilon\left(\gamma^{4}-\gamma^{2}\right)+\gamma^{2}\right]}{\gamma^{2}\left(\gamma-1\right)}\left(1+2\epsilon\right)\\
&\leq \frac{1+2\epsilon}{\gamma\left(\gamma-1\right)}r_{i}+\left[\epsilon\left(\gamma^{2}-1\right)+1\right]\left(1+2\epsilon\right)\\
&\leq  \frac{\gamma}{\gamma\left(\gamma-1\right)}r_{i}-\left(\frac{\gamma-\left(1+2\epsilon\right)}{\gamma\left(\gamma-1\right)}\right)r_{i}+\left[\epsilon\left(\gamma^{2}-1\right)+1\right]\left(1+2\epsilon\right)\\
&\leq \frac{r_{i}}{\gamma-1}-\frac{r_{i}}{2\gamma}+\left[\epsilon\left(\gamma^{2}-1\right)+1\right]\left(1+2\epsilon\right).
\end{align*}
The equation in the second line follows from the recursive definition
of $r_{i}$; the last inequality follows from $\gamma>1+4\epsilon$.
Now, using that $r_{i}>2\gamma\left[\epsilon\left(\gamma^{2}-1\right)+1\right]\left(1+2\epsilon\right)$
for all $i$, we have that $Rev^{(2)}\left(t_i,C_{i+1}\right)< Rev^{(2)}\left(t_i,D_{i}\right)$.
\item What happens if we set $q_{i}=D_{i+1}$?
\begin{align*}
Rev^{(2)}\left(t_i,D_{i+1}\right) &=  D_{i+1}\cdot \bar{F}_{i}\left(D_{i+1}\right)\\
&\leq\frac{\gamma r_{i+1}}{\left(\gamma-1\right)h_{i+1}}h_{i+1}\left(2-1/\gamma\right)\\
&\leq \frac{r_{i+1}}{\left(\gamma-1\right)}\left(2\gamma-1\right)\\
&=\frac{2\gamma-1}{\gamma^{3}}\cdot\frac{\gamma r_{i}+\left(\gamma-1\right)\left[\epsilon\left(\gamma^{4}-\gamma^{2}\right)+\gamma^{2}\right]}{\gamma^{2}\left(\gamma-1\right)}\\
&\leq \frac{\gamma r_{i}+\left(\gamma-1\right)\left[\epsilon\left(\gamma^{4}-\gamma^{2}\right)+\gamma^{2}\right]}{\gamma^{2}\left(\gamma-1\right)}\\
&\leq Rev^{(2)}\left(t_i,D_{i}\right),
\end{align*}
where the last inequality follows from the analysis for $Rev^{(2)}\left(t_i,C_{i+1}\right)$.
\item What about the revenue when we charge $C_{i+2}$, $D_{i+2}$ ? We
reduce this case to what we already know about the revenue from type
$t_{i+1}$:

Observe that $\bar{F}_{i}\left(C_{i+1}\right)>\bar{F}_{i+1}\left(C_{i+1}\right)$,
but $\bar{F}_{i}\left(C_{i+2}\right)<\bar{F}_{i+1}\left(C_{i+2}\right)$. Therefore,
\[
Rev^{(2)}\left(t_i,C_{i+2}\right)< Rev^{(2)}\left(t_{i+1},C_{i+2}\right)\leq Rev^{(2)}\left(t_{i+1},C_{i+1}\right)< Rev^{(2)}\left(t_i,C_{i+1}\right)\mbox{ .}
\]
A similar argument works for $D_{i+2}$, and the claim follows by
induction for all $C_{j},D_{j}$.

\item Finally, for points $x>D_{n+1}$, we will show that $Rev^{(2)}\left(t_n,D_{n+1}\right)$
is greater than $Rev^{(2)}\left(t_n,x\right)$, and the claim will
follow for all $i\leq n$ by the previous argument (Recall that in
the domain $x>D_{n+1}$, $\bar{F}_{i}$ is the same for all $i$).

$\bar{F}_{n}$ changes its values at points $2^{k}D_{n+1}$. We have:
\begin{gather*}
Rev^{(2)}\left(t_n,2^{k}D_{n+1}\right)=2^{k}D_{n+1}\cdot \bar{F}_{n}\left(2^{k}D_{n+1}\right)=\frac{D_{n+1}h_{n+1}}{2\gamma}<\frac{D_{n}h_{n}}{2\gamma}= \frac{Rev^{(2)}\left(t_n,D_{n}\right)}{2}. \qedhere
\end{gather*}
\end{itemize}
\end{proof}

\subsection{Putting it all together}

In Lemma~\ref{lem:completeness} we saw that if there exists an independent set of size $|S|$ there exists an IC and IR satisfying pricing which yields revenue

\[ (1-p) Rev(t^*,P^*,Q^*) + p \sum_{i\in V} w_i Rev(t_i,A_i,D_i) + pr|S|. \]

In Lemma~\ref{lem:soundness} we saw that any IC and IR satisfying pricing cannot yield more revenue than
\[ \left(1-p\right) Rev\left(t^*,P^{*},Q^{*}\right)+p\sum_{i\in V} w_i Rev(t_i,B_{i},D_{i})+pr\left|S\right|, \]
where $|S|$ is the size of the maximum independent set in $G$. 

All that's left is to show that a graph with maximum independent set of size $|S|-1$ cannot yield revenue $(1-p) Rev(t^*,P^*,Q^*) + p \sum_{i\in V} w_i Rev(t_i,A_i,D_i) + pr|S|$. 
To this end we need to show that,

\begin{align*}
& (1-p) Rev(t^*,P^*,Q^*)+p \sum_{i\in V} w_i Rev(t_i,A_i,D_i) + pr|S| > \\
& \;\;\;\;\;\;\;\; \;\;\;\;\;\;\;\; \;\;\;\;\;\;\;\; (1-p) Rev(t^*,P^*,Q^*)+p\sum_{i\in V} w_i Rev(t_i,B_{i},D_{i})+pr(\left|S\right|-1). 
\end{align*}
or equivalently,
\begin{align*}
& pr > p\sum_{i\in V} w_i (Rev(t_i,B_{i},D_{i}) - Rev(t_i,A_i,D_i)) = p\sum_{i\in V} w_i \epsilon \\
& \iff r > \sum_{i\in V} \frac{r \epsilon}{r_i} \\
& \iff 1 > \sum_{i\in V} \frac{\epsilon}{r_i},
\end{align*}

which is true since $\epsilon = \frac{1}{n^2}$, and each $r_i = O(n)$. With this the reduction is complete.

\section{\NP-hardness: Completeness}
\label{app:reduction}


In this section we show that any independent set $S$ in $G$ corresponds to a feasible pricing in our mechanism: $(B_i, C_i)$ for $i\in S$, $(A_j, D_j)$ for $j \notin S$, and $(P^*, Q^*)$ for type $t^*$. 

\begin{lem}
\label{lem:completeness}
Let $S$ be an independent set of $G$. There exists a pricing for our mechanism that satisfies IC and IR and achieves revenue:

$$ (1-p) Rev(t^*,P^*,Q^*) + p \sum_{i\in V} w_i Rev(t_i,A_i,D_i) + pr|S| $$
\end{lem}

We first show that the IC constraints are satisfied between any pair of types $t_i$ and $t_j$ that are not both charged $(B_i, C_i)$
 - edge or no edge in the graph (Claim~\ref{lem:all-but-BCBC}). 
Then, we show that the IC constraints are satisfied between type $t^*$ and type $t_i$, for any $i \in [n]$ (Claim~\ref{lem:star-allows-all}). 
Finally we prove that charging $(B_i,C_i)$ and $(B_j,C_j)$ does not violate the IC constraints if there is no $(i,j)$ edge in the graph (Claim~\ref{lem:noedge-BCBC}).

\begin{claim}
\label{lem:all-but-BCBC}
Charging types $t_i$ and $t_j$, for $j>i$, doesn't violate the IC constraints between $t_i$ and $t_j$ for any of the pairs $(B_i,C_i)$/$(A_j,D_j)$, $(A_i,D_i)$/$(B_j,C_j)$ or $(A_i,D_i)$/$(A_j,D_j)$.
\end{claim}

\begin{proof}

We need to show that all the following are always true:
\begin{enumerate}
\item $$ \int_{C_i}^{D_j} \bar{F}_i(x) dx \geq B_i - A_j \geq \int_{C_i}^{D_j} \bar{F}_j(x) dx $$
\item $$ \int_{D_i}^{D_j} \bar{F}_i(x) dx \geq A_i - A_j \geq \int_{D_i}^{D_j} \bar{F}_j(x) dx $$
\item $$ \int_{D_i}^{C_j} \bar{F}_i(x) dx \geq A_i - B_j \geq \int_{D_i}^{C_j} \bar{F}_j(x) dx $$
\end{enumerate}

It follows from Table~\ref{table:integrals} that the left hand sides hold.
For the right hand sides, first notice that $\bar{F}_j$ is always lower than $\bar{F}_i$ in the intervals we're interested in. The first inequality is tight for $\bar{F}_i$, thus $ \int_{C_i}^{D_j} \bar{F}_j(x) \leq B_i - A_j$. For $(A_i,D_i)$/$(B_j,C_j)$ and $(A_i,D_i)$/$(A_j,D_j)$ we will use induction:

\begin{itemize}
\item Basis $j=i+1$: 
\begin{equation*}
\begin{split}
\int_{D_i}^{C_{i+1}} \bar{F}_{i+1}(x) dx & = (C_{i+1}-D_i)h_{i+1} = (C_{i+1}-D_i) \frac{h_i}{\gamma^4} \\
& = \frac{1 - \epsilon}{\gamma^2} < 1-\epsilon = A_i - B_{i+1}
\end{split}
\end{equation*}

And:
\begin{equation*}
\begin{split}
\int_{D_i}^{D_{i+1}} \bar{F}_{i+1}(x) dx & = (C_{i+1}-D_i)h_{i+1} + (D_{i+1}-C_{i+1})\frac{h_{i+1}}{\gamma} \\
& = \frac{1 - \epsilon}{\gamma^2} + \frac{\epsilon}{\gamma} < 1 = A_i - A_{i+1}
\end{split}
\end{equation*}

\item For $j$ we have the following:
\begin{equation*}
\begin{split}
\int_{D_i}^{C_j} \bar{F}_j (x) dx & \leq \int_{D_i}^{D_{j-1}} \bar{F}_{j-1}(x) dx + \int_{D_{j-1}}^{C_j} \bar{F}_j(x) dx \\
& \leq (A_i - A_{j-1}) + (A_{j-1} - B_j) = A_i - B_j
\end{split}
\end{equation*}

and

\begin{equation*}
\begin{split}
\int_{D_i}^{D_j} \bar{F}_j (x) dx & \leq \int_{D_i}^{D_{j-1}} \bar{F}_{j-1}(x) dx + \int_{D_{j-1}}^{D_j} \bar{F}_j(x) dx \\
& \leq A_i - A_{j-1} + A_{j-1} - A_j = A_i - A_j
\end{split}
\end{equation*}

\end{itemize}

\end{proof}

\begin{claim}
\label{lem:star-allows-all}
When type $t^*$ is charged $(P^*,Q^*)$, charging $t_i$ the pair $(B_i,C_i)$ or the pair $(A_i,D_i)$ doesn't violate the IC constraints between $t_i$ and $t^*$.
\end{claim}

\begin{proof}
The IC constraints between $t_i$ and $t^*$ are either 
$$ \int_{C_i}^{Q^*} \bar{F}_i(x) dx \geq B_i - P^* \geq \int_{C_i}^{Q^*} \bar{F}_*(x) dx $$

or 

$$ \int_{D_i}^{Q^*} \bar{F}_i(x) dx \geq A_i - P^* \geq \int_{D_i}^{Q^*} \bar{F}_*(x) dx $$

In both cases, the inequalities can be verified easily using Table~\ref{table:integrals}. 

\end{proof}

\begin{claim}
\label{lem:noedge-BCBC}
If $(i,j) \not\in E$ the charging type $t_i$ the pair $(B_i,C_i)$ and type $t_j$ the pair $(B_j,C_j)$ doesn't violate the IC constraints between $t_i$ and $t_j$.
\end{claim}

\begin{proof}

The IC constraint between $t_i$ and $t_j$ for this pricing is:
$$ \int_{C_i}^{C_j} \bar{F}_i(x) dx \geq B_i - B_j \geq \int_{C_i}^{C_j} \bar{F}_j(x) dx $$

\begin{itemize}
\item $j=i+1$: $\int_{C_i}^{C_{i+1}} \bar{F}_i(x) dx = \int_{C_i}^{D_i} \bar{F}_i(x) dx + \int_{D_i}^{C_{i+1}} \bar{F}_i(x) dx$. The first term is equal to $\frac{\epsilon}{\gamma}$, and when there is no $(i,i+1)$ edge, the second term is equal to $1 - \frac{\epsilon}{\gamma}$, thus the left hand side is immediate. The right hand side is satisfied trivially, since $\bar{F}_{i+1}$ is always below $\bar{F}_i$ between $C_i$ and $C_{i+1}$ and $\bar{F}_i$ gives a tight constraint.

\item $j>i+1$: Again, $ \int_{C_i}^{C_j} \bar{F}_i(x) dx = \int_{C_i}^{D_{j-1}} \bar{F}_i(x) dx + \int_{D_{j-1}}^{C_j} \bar{F}_i(x) dx$. From Table~\ref{table:integrals} we can see that the first term is always $j-1-i+\epsilon$, and the second term is $1-\epsilon$ when $(i,j)\not\in E$.

For the right hand side we have $\int_{D_i}^{C_j} \bar{F}_j(x) dx \leq A_i - B_j$ from Claim ~\ref{lem:all-but-BCBC}. Since $\bar{F}_j$ is below $\bar{F}_i$ between $C_i$ and $D_i$, and $\int_{C_i}^{D_i} \bar{F}_i(x) dx = \frac{\epsilon}{\gamma} < \epsilon$ we get that:
\begin{equation*}
\begin{split}
\int_{C_i}^{C_j} \bar{F}_j(x) dx & = \int_{C_i}^{D_i} \bar{F}_j(x) dx + \int_{D_i}^{C_j} \bar{F}_j(x) dx \\
& < \int_{C_i}^{D_i} \bar{F}_i(x) dx + A_i - B_j \\
& < \epsilon + A_i - B_j = B_i - B_j
\end{split}
\end{equation*}

\end{itemize}
\end{proof}

\section{Given first-stage prices, deterministic mechanisms are easy}
\label{app:given-first-day}

Here we include a proof of Theorem \ref{thm:FPTAS-unimodular}.

\FPTASunimodular*

This result shows us something very important about the structure of hard instances and what a possible reduction can look like: the mechanism gadgets cannot have fixed prices for one of the two stages; variation on both stages is required.

\begin{proof}
Our problem is to find optimal second stage prices, when we have committed to charging every first stage type $t_i$, $i \in [n]$, a payment of $p_i$ in the first stage.
For now, assume that all types are allocated the item on the first stage (this fact will not be used in a crucial way, and an almost identical algorithm works). A first question is whether there even exist second stage prices that do not violate the IC constraints. Then, if there exist such prices, how would we optimize over them in order to maximize the seller's revenue? As it turns out, these sub-problems are easy; we can construct an FPTAS using an integer program.

It will be useful to think about the incentive constraints as follows: given first stage prices $p_i$, $p_j$ and a second stage price $q_i$ the incentive constraints between types $t_i$ and $t_j$ give a certain interval in which $q_j$ is allowed to be. Specifically, for $p_i > p_j$, we can define a lower bound $lb_{i,j} ( p_i,p_j,q_i )$ and an upper bound $ub_{i,j} ( p_i,p_j,q_i )$ for $q_j$ as follows:

$$ lb_{i,j} ( p_i,p_j,q_i ) = min_{q \in [q_i,q_{max}]} \{ \int_{q_i}^q \bar{F}_i(x) dx \geq p_i - p_j  \}, $$
$$ ub_{i,j} ( p_i,p_j,q_i ) = max_{q \in [q_i,q_{max}]} \{ \int_{q_i}^q \bar{F}_j(x) dx \leq p_i - p_j \}, $$

where $q_{max}$ is the maximum value in the support of $V^{(2)}$. 
The integer program works as follows: first discretize the second stage, i.e. we only consider prices of the form $k \epsilon$ for some $\epsilon>0$ and $k \in [m]$, where $m$ is an integer large enough so that $\epsilon \cdot m$ is larger than the maximum value in the support of $V^{(2)}$. We have a binary variable $x_i^k$ for each type $t_i$ and price $k \epsilon$, such that if $x_i^k = 1$ the second stage price is greater than $k \epsilon$. Also, we have a constant $a_i^k$ for the revenue of charging type $t_i$ price $k \epsilon$ ($a_i^k$ is easily calculated from the input distribution). Since $p_i$ and $p_j$ are given, we simply write $lb_{i,j}(k)$ and $ub_{i,j}(k)$. The integer program is as follows:

	\[
	\begin{array}{lrclr}
	\textrm{max}   & \sum_{i \in [n]} \sum_{k \in [m]} x_i^k ( a_i^k - a_i^{k-1} ) &      &   & \\
	\textrm{subject to} & x_i^k    &\leq    & x_j^{ lb_{ij} (k) } &\forall i,j \in [n], k \in [m] \\
			    	     & x_j^{ ub_{ij} (k) + 1 }   &\leq & x_i^{k+1}  &\forall i,j \in [n], k \in [m-1]\\
				     & x_i^k &\leq & x_i^{k-1} & \forall i \in [n], k = 1, \dots, m\\
				     & x_i^k \in \{ 0,1 \} &\forall i \in [n], k \in [m] &
	\end{array}
	\]

The first two constraints encode incentive compatibility, by guaranteeing that $q_j \in [lb_{ij}(q_i), ub_{ij}(q_i)]$ for all $i,j$: (1) if $i$ is charged at least $k\epsilon$, then $j$ is charged at least $lb_{ij}(k)$, (2) if $j$ is charged at least $ub_{ij}(k)$, then $i$ is charged at least $k \epsilon$ (or, equivalently, if $i$ is charged at most $k\epsilon$, then $j$ is charged at most $ub_{ij}(k)$). The third constraint simply encodes that $x_i^k$ is non-increasing.

Observe that the constraints matrix is totally unimodular: every entry is $0, +1$ or $-1$, and every row has at most two non-zero entries with different signs \cite{papadimitriou1982combinatorial}. Thus the relaxation gives us an integer solution \cite{papadimitriou1982combinatorial}.
\end{proof}
	
\section{Deterministic Mechanisms with Independent stages}
\label{app:independent}
\input{independent}

\section{Ex-post IR implies Stage-wise ex-post IR}
\label{app:stage-wise}

In this section we show how, given a mechanism that is ex-post IR, we can get a mechanism that is stage-wise ex-post IR, that remains truthful, and guarantees at least as much expected revenue.

Let $M$ be an ex-post IR mechanism, that in each stage $d$ takes as input the history of reported valuations $h^{[d-1]}$, mechanism outcomes $\omega^{[d-1]}$ and current reports $v^{[d]}$, and outputs a distribution over outcomes, where outcome $\theta$ occurs with probability $Pr[\theta]$, allocates the item to buyer $i$, which we indicate by writing $x_i(h^{[d-1]}; \omega^{[d-1]}; v^{[d]}; \theta) = 1$, and charges $p_i(h^{[d-1]}; \omega^{[d-1]}; v^{[d]}; \theta)$ (the mechanism is allowed to charge buyers that didn't receive the item).

Let $M'$ be the mechanism the same exact mechanism, with the only difference that the payment $p'_i(h^{[d-1]}; \omega^{[d-1]}; v^{[d]}; \theta)$ is equal $v^{(d)}_i$ if buyer $i$ got the item, and zero otherwise, for all stages $d = 1, \dots, D-1$. In the last stage, the allocation of $M'$ remains the same as $M$, but the payment is instead $p'_i(h^{[D-1]}; \omega^{[D-1]}; v^{[D]}; \theta) = \sum_{d=1}^{D} p_i(h^{[d]}; \omega^{[d-1]}; v^{[d]}; \omega_d) - \sum_{d=1}^{D-1} p'_i(h^{[d]}; \omega^{[d-1]}; v^{[d]}; \omega_d)$, where we overload notation, and use $h^{[d]}$ (resp. $\omega^{[d]}$, $v^{[d]}$) for the restriction of $h^{[D-1]}$ to the first $d$ stages, and $\omega_d$ is the outcome (some $\theta$) in stage $d$ according to $\omega^{[D-1]}$.

By definition, the expected revenue of $M$ and $M'$ is the same: (part of) the payments in $M$ simply get ``pushed'' to the last stage in all possible outcomes. Also, $M'$ is clearly stage-wise ex-post IR for the first $D-1$ stages. For the last stage, consider an outcome where buyer $i$ gets the item (the case that buyer $i$ doesn't get the item is identical): the buyer's utility is 
\begin{multline*}
v^{(D)}_i - \sum_{d=1}^{D} p_i(h^{[d]}; \omega^{[d-1]}; v^{[d]}; \omega_d) + \sum_{d=1}^{D-1} p'_i(h^{[d]}; \omega^{[d-1]}; v^{[d]}; \omega_d) \\ = \sum_{d=1}^D v^{(d)}_i \cdot I\{ \text{$i$ got item $d$ according to $\omega^{[D]}$} \} - \sum_{d=1}^{D} p_i(h^{[d]}; \omega^{[d-1]}; v^{[d]}; \omega_d),
\end{multline*}
where $I\{.\}$ is the indicator function. Notice that the RHS is precisely the ex-post utility of $i$ in the last stage in $M$, according to $\omega^{[D-1]}$, $h^{[D-1]}$ and the outcome in the last stage, and thus it is non-negative.

Finally, for incentive compatibility, we show the single buyer case, for multiple buyers and D-DIC or B-DIC, the proof is identical. Consider the expected utility of the buyer in stage $d$, when her true value is $v$, public and private histories are $h^{[d]}$ and $\mathrm{\bar{h}}_i^{[d]}$ and outcomes are $\omega^{[d]}$. Since each stage utility is zero, her expected utility is just the expected utility of the last stage 
\begin{align}
\mathbb{E}_{h^{[d+1:D]}, \omega^{[d+1:D]} | \mathrm{\bar{h}}_i^{[d]}} [ v_D x_D(.;v) - \sum_{t=1}^D p_t(.;v) + \sum_{t=1}^{D-1} p'_t(.;v) ],  \label{eq:1}
\end{align}
where (1) we have shortened notation to $p_t(.;v)$ (resp. $x_t(.;v)$, $p'_t(.;v)$) to focus on a stage $t$ and how the deviation in stage $d$ affects it, and (2) $x_t(.;v)$ is the expected allocation in stage $t$ (we again overload notation). When misreporting to $v'$, her expected utility is $v \cdot x_{d}(.;v') - p'_d(.;v')$ from this stage (where $x_d(.;v')$ is the expected allocation in stage $d$ when the report is $v'$, and $p'_d(.;v')$ is the expected payment), zero in all stages $d+1$ through $D-1$ (since the buyer behaves truthfully then), plus $\mathbb{E}_{h^{[d+1:D]}, \omega^{[d+1:D]}|\mathrm{\bar{h}}_i^{[d]}} [ v_D x_D(.;v') - \sum_{t=1}^D p_t(.;v') + \sum_{t=1}^{D-1} p'_t(.;v') ]$ from the last stage. We would like for equation~\ref{eq:1} to be at least this expression. Since the terms $p'_t(.;v)$ and $p'_t(.;v')$, as well as $p_t(.;v)$ and $p_t(.;v')$, are identical for $t \leq d-1$ these terms cancel out. Furthermore, $p'_t(.;k) = k \cdot x_t(.;k)$. We thus want to show that
\begin{align*}
\mathbb{E}[ \sum_{t=d}^D v^{(t)} x_t(.;v) - \sum_{t=d}^D p_t(.;v) ] & \geq vx_d(.;v') + \mathbb{E} [ \sum_{t=d+1}^D v^{(t)} x_t(.;v') - \sum_{t=d}^D p_t(.;v') ],
\end{align*}
or equivalently
\begin{multline*}
v x_d(.;v) - p_d(.;v) + \mathbb{E}[ \sum_{t=d+1}^D v^{(t)} x_t(.;v) - \sum_{t=d+1}^D p_t(.;v) ]  \\ \geq vx_d(.;v') - p_d(.;v') + \mathbb{E} [ \sum_{t=d+1}^D v^{(t)} x_t(.;v') - \sum_{t=d+1}^D p_t(.;v') ],
\end{multline*}
where the $d$-th terms can be taken out of the expectation since the expectation is with respect to the events after stage $d$. Notice that this is the IC constraint for $M$ and is therefore satisfied.

\section{Construction of \texorpdfstring{$D_{1}$}{Lg} and \texorpdfstring{$D_{2}$}{Lg}: Proof of Lemma \ref{lem:D_1-and-D_2}}
\label{app:no_contract}

\begin{proof}
We explicitly define $D_{1}$ and $D_{2}$, and then check that they
satisfy all the requirements. We use $D\left(v\right)$ to denote
the probability that distribution $D$ assigns to value $v$. Let
$O\left(1/\ln\left(k\right)\right)\leq\alpha<1/5$ and $1/2\leq\beta\leq1$
be parameters to be defined soon. We define the first distribution
as follows:
\[
D_{1}\left(v\right)=\begin{cases}
1-\alpha & v=0\\
\left(\alpha\cdot\frac{1}{2}\right)-2\epsilon^{2} & v=1+\epsilon\\
\alpha\cdot\left(\frac{1}{v}-\frac{1}{v+1}\right) & v\in\left\{ 2,\dots,k\right\} \\
\left(\alpha\cdot\frac{1}{v}\right)+2\epsilon^{2} & v=k
\end{cases}
\]

\noindent Notice that prices $1+\epsilon$ and $k$ have probabilities higher
than the equal-revenue curve for $v\in\left\{ 2,\dots,k\right\} $;
one of them will always be optimal. Similarly, we let 
\[
D_{2}\left(v\right)=\begin{cases}
1-\beta & v=0\\
\left(\beta\cdot\frac{1}{2}\right)+\left(\frac{k}{2}+1\right)\epsilon^{2} & v=1\\
\beta\cdot\left(\frac{1}{2}-\frac{1}{k}\right)-\frac{k}{2}\epsilon^{2} & v=2\\
\left(\beta\cdot\frac{1}{v}\right)-\epsilon^{2} & v=k
\end{cases}
\]
Price $1$ has relatively high probability, and price $k$ comes after;
thus for $D_{2}$ alone price $1$ will be optimal, but together with
$D_{1}$, price $k$ maximizes the revenue.

\paragraph{Myerson pricing} For prior $D_{1}$ the maximal revenue is achieved
by $p_{2}=1+\epsilon$: 
\begin{equation}
\forall p'\in\{2,\dots,k\}\,\,\,\left(1+\epsilon\right)\cdot\Pr_{v_{2}\sim D_{1}}\left[v_{2}\geq1+\epsilon\right]=\alpha\left(1+\epsilon\right)>\alpha+2p'\epsilon^{2}=p'\cdot\Pr_{v_{2}\sim D_{1}}\left[v_{2}\geq p'\right].\label{eq:rev-D1}
\end{equation}
Similarly, for $D_{2}$ the revenue is maximized by $p_{2}=1$:
\begin{equation}
\forall p'\in\{2,k\}\,\,\,1\cdot\Pr_{v_{2}\sim D_{2}}\left[v_{2}\geq1\right]=\beta>\beta-k\epsilon^{2}=p'\cdot\Pr_{v_{2}\sim D_{1}}\left[v_{2}\geq p'\right].\label{eq:rev-D2}
\end{equation}
For the seller's initial prior, $\frac{1}{2}D_{1}+\frac{1}{2}D_{2}$,
the revenue is maximized by $k$:
\begin{equation}
k\cdot\Pr_{v_{2}\sim\frac{1}{2}D_{1}+\frac{1}{2}D_{2}}\left[v_{2}\geq k\right]=\frac{1}{2}\alpha+\frac{1}{2}\beta+\frac{k}{2}\epsilon^{2}>\frac{1}{2}\alpha+\frac{1}{2}\beta=1\cdot\Pr_{v_{2}\sim\frac{1}{2}D_{1}+\frac{1}{2}D_{2}}\left[v_{2}\geq1\right]\mbox{.}\label{eq:rev-D1-D2}
\end{equation}
Finally, we show that for any convex combination $\lambda D_{1}+\left(1-\lambda\right)D_{2}$
of the distributions, the revenue is maximized by some price $p_{2}\in\{1,1+\epsilon,k\}$.
It is easy to see that the optimal price belongs to the support
of the mixed distribution. Yet, for any $p'\in\{3,\dots,k-1\}$ we
have:
\begin{align*}
k\cdot\Pr_{v_{2}\sim\gamma D_{1}+\left(1-\gamma\right)D_{2}}\left[v_{2}\geq k\right] &=  \gamma\alpha+\left(1-\gamma\right)\beta+\left(3\gamma-1\right)\cdot k\epsilon^{2}\\
& > \gamma\alpha+\left(1-\gamma\right)\beta\cdot\frac{p'}{k}+\left(3\gamma-1\right)p'\epsilon^{2}\\
& = p'\cdot\Pr_{v_{2}\sim\gamma D_{1}+\left(1-\gamma\right)D_{2}}\left[v_{2}\geq p'\right].
\end{align*}
Similarly, for $p'=2$ and $\gamma<1$, we still have 
\[
k\cdot\Pr_{v_{2}\sim\gamma D_{1}+\left(1-\gamma\right)D_{2}}\left[v_{2}\geq k\right]>\gamma\alpha+\left(1-\gamma\right)\beta+\left(6\gamma-k\right)\epsilon^{2}=2\cdot\Pr_{v_{2}\sim\gamma D_{1}+\left(1-\gamma\right)D_{2}}\left[v_{2}\geq2\right]\mbox{.}
\]

\paragraph{Buyer's utility} 
Recall that $u_{2}\left(D'\mid D\right)$ denotes
the buyer's expected utility from the second-stage mechanism when her
true distribution is $D$, but the seller runs a Myerson mechanism
against a (possibly misreported) prior of $D'$.

\begin{description}
\item [{Truthfulness}] When the buyer draws her second stage valuation from
$D_{1}$ and the price is $p\left(D_{1}\right)=1+\epsilon$, her utility
is 
\[
u_{2}\left(D_{1}\mid D_{1}\right)=\alpha\cdot\left(H_{k}-1-\epsilon/2\right)+2\epsilon^{2}\cdot\left(k-1-\epsilon\right)
\]
together with a discount of $\epsilon/5$ on the first stage, it is
greater than the utility from price $p\left(D_{2}\right)=1$: 
\[
u_{2}\left(D_{2}\mid D_{1}\right)=\alpha\cdot\left(H_{k}-1+\epsilon/2\right)+2\epsilon^{2}\cdot\left(k-1-\epsilon\right)=u_{2}\left(D_{1}\mid D_{1}\right)+\alpha\epsilon\mbox{.}
\]
On the other hand, if the buyer draws her valuation from $D_{2}$,
then we have:
\[
u_{2}\left(D_{1}\mid D_{2}\right)=\left(\beta\left(\frac{1}{k}\right)-\epsilon^{2}\right)\cdot\left(k-1-\epsilon\right)+\left(\beta\left(\frac{1}{2}-\frac{1}{k}\right)-\frac{k}{2}\epsilon^{2}\right)\cdot\left(1-\epsilon\right)\mbox{;}
\]
 as well as
\begin{align*}
u_{2}\left(D_{2}\mid D_{2}\right) &=  \left(\beta\left(\frac{1}{k}\right)-\epsilon^{2}\right)\cdot\left(k-1\right)+\left(\beta\left(\frac{1}{2}-\frac{1}{k}\right)-\frac{k}{2}\epsilon^{2}\right)\\
 &=  u_{2}\left(D_{1}\mid D_{2}\right)+\left(\left(\beta\cdot\frac{1}{2}\right)-\left(\frac{k}{2}+1\right)\epsilon^{2}\right)\cdot\epsilon\mbox{.}
\end{align*}
 Therefore, the discount must satisfy $\left(\beta\cdot\frac{1}{2}-\left(\frac{k}{2}+1\right)\epsilon^{2}\right)\epsilon>\epsilon/5>\alpha\epsilon$.
\item [{Value of OTR}]
The value of the OTR for a buyer with prior $D_{1}$
is given by 
\begin{align*}
u_{2}\left(D_{1}\mid D_{1}\right)&+\epsilon/5-u_{2}\left(\frac{1}{2}D_{1}+\frac{1}{2}D_{2}\mid D_{1}\right) \\
&= \alpha\cdot\left(H_{k}-1-\epsilon/2\right)-2\epsilon^{2}\cdot\left(k-1-\epsilon\right)+\epsilon/5\\
 &= \alpha\cdot\left(H_{k}-1\right)+O\left(\epsilon\right)\mbox{.}
\end{align*}
The value for a buyer with prior $D_{2}$ is given by
\begin{align*}
u_{2}\left(D_{2}\mid D_{2}\right)&-u_{2}\left(\frac{1}{2}D_{1}+\frac{1}{2}D_{2}\mid D_{2}\right) \\
&=  \left(\beta\left(\frac{1}{k}\right)-\epsilon^{2}\right)\cdot\left(k-1\right)+\left(\beta\left(\frac{1}{2}-\frac{1}{k}\right)-\frac{k}{2}\epsilon^{2}\right)\\
&=  \beta\left(1.5-O\left(\frac{1}{k}\right)\right)+O\left(\epsilon^{2}\right)\mbox{.}
\end{align*}
Finally, in order to achieve equal value of the OTR, choose $\alpha$
and $\beta$ such that 

\begin{multline*}
\alpha\cdot\left(H_{k}-1-\epsilon/2\right)-2\epsilon^{2}\cdot\left(k-1-\epsilon\right)+\epsilon/5 \\ =\left(\beta\left(\frac{1}{k}\right)-\epsilon^{2}\right)\cdot\left(k-1\right)+\left(\beta\left(\frac{1}{2}-\frac{1}{k}\right)-\frac{k}{2}\epsilon^{2}\right).
\end{multline*}

In particular consider $\beta=1$ and $\alpha\approx1.5/H_{k}$.
\end{description}

\paragraph{Seller's revenue} As we already showed in (\ref{eq:rev-D1})-(\ref{eq:rev-D1-D2}),
the optimal expected revenue from the second item is approximately
the same whether the seller learns the buyer's partial information
or not:
\[
\frac{1}{2}\textup{\texttt{Rev}}\left(D_{1}\right)+\frac{1}{2}\textup{\texttt{Rev}}\left(D_{2}\right)=\frac{1}{2}\alpha\left(1+\epsilon\right)+\frac{1}{2}\beta=\frac{1}{2}\alpha+\frac{1}{2}\beta+O\left(\epsilon\right)
\]
versus
\[
\textup{\texttt{Rev}}\left(\frac{1}{2}D_{1}+\frac{1}{2}D_{2}\right)=\frac{1}{2}\alpha+\frac{1}{2}\beta+\frac{k}{2}\epsilon^{2}=\frac{1}{2}\alpha+\frac{1}{2}\beta+O\left(\epsilon^{2}\right)\mbox{.} \qedhere
\]
\end{proof}

\end{appendix}

%% file: independent.tex

\THMindependent*

As we observed in Section~\ref{app:separations}, the two-stage optimal deterministic mechanism can be rather bizarre, even when the distributions are independent. 
Nonetheless, we show below that when the distributions are independent, the optimal mechanism satisfies some strong structural properties, which in turn significantly reduce our search space.

The curious reader might be wondering whether the promised utility framework, originally introduced by \cite{green1987lending,spear1987repeated,thomas1990income}, can help us resolve this question. This framework provides an approach for designing mechanisms that are dynamic incentive compatible when values are independent over time (this approach is complex in the presence of history dependence~\cite{fernandes2000recursive}). The approach typically involves solving the problem recursively via a dynamic program, and typically the state space grows exponentially (in the input, e.g. number of stages and number of buyers). This so-called curse of dimensionality has no bite here, since we only consider a single buyer and two stages. Part of our analysis here is ``forward-looking'' (the optimal choice for period $1$ is taken, given a choice for period $2$), similar to the promised utility framework, therefore, even though we are not aware of a general way to use this framework to get optimal \emph{deterministic} mechanisms, we cannot rule out the possibility that this approach can be used here as well.

Going back to the proof of Theorem~\ref{thm:independent}, we should first decide who gets the item on the first stage. 

\begin{claim} [First-stage allocation monotonicity] \label{cla:ind_alloc-monotonicity}
There exist an optimum mechanism such that $v_i > v_j \implies x_1(t_i) \geq x_1(t_j)$.
\end{claim}
\begin{proof}
Assume that is not the case, and there exist $i$ and $j$, $v_i > v_j$, such that if the buyer reports $t_j$ on the first stage she is allocated the item for a payment of $p_j$ and a second stage promised price $q_j$, but if she reports $t_i$, she is not allocated the item in the first stage, and is promised a second stage price $q_i$. A buyer with type $t_j$ doesn't want to report $t_i$, thus $v_j - p_j + E_{v_2 \sim D_2}[ max\{ v_2 - q_j , 0\} ] \geq  E_{v_2 \sim D_2}[ max\{ v_2 - q_i , 0\} ]$. Similarly, type $t_i$ doesn't want to report $t_j$, thus $ E_{v_2 \sim D_2}[ max\{ v_2 - q_i , 0\} ] \geq v_i - p_j +  E_{v_2 \sim D_2}[ max\{ v_2 - q_j , 0\} ]$. The two inequalities imply that $v_j \geq v_i$, a contradiction.
\end{proof}

Henceforth we say that $t_i$ is a {\em winning} type if $x_1(t_i)=1$ and {\em losing} otherwise. There are $n+1$ possible first stage allocations: no one gets the item, only the highest type gets the item, the highest two types get the item, and so on. Our algorithm tries all of them. Therefore, our problem is reduced to finding the optimal mechanism for a given subset of losing types.

Observe that our IC constraints between winning types $t_i, t_j$ reduce to:
\begin{gather}
 p(t_i) - p(t_j) = \int_{q(t_i)}^{q(t_j)} \bar{F}(x) dx.  \label{eq:ind-IC-win_win}
\end{gather}

Having a tight equality means that if we know three of $\big\{p(t_i), p(t_j), q(t_i), q(t_j)\big\}$, 
we can immediately compute the fourth. 
The IC constraints between a winning $t_i$ and a losing $t_j$ are:
\begin{gather}
 v_j - p(t_i) \leq \int_{q(t_j)}^{q(t_i)} \bar{F}(x) dx \leq t_i - p(t_i). \label{eq:ind-IC-win_lose}
\end{gather}

The following observation is immediate from the IC constraints, and it will be useful in proving the rest of the structural claims:

\begin{obs}\label{obs:1}
Take any truthful mechanism, and change the prices only for type $t_i$, such that the utility for type $t_i$ does not change. Then the mechanism remains truthful.
\end{obs}

Now, finding two of the three unknown prices becomes much easier thanks to the following claim:

\begin{claim} \label{cla:ind_first-day_supp}
There exist an optimum mechanism that satisfies Claim \ref{cla:ind_alloc-monotonicity}, and such that for any winning type $t_i$ either: $p(t_i) = v_i$; or $q(t_i) = 0$. 
\end{claim}

\begin{proof}
Let $q_{next}$ be the maximum point in the support of the second stage distribution such that $q_{next} < q(t_i)$,
if such a point exists, and $0$ otherwise. 
Suppose that for any $\epsilon > 0$, $p(t_i) \leq v_i-\epsilon$ and $q(t_i) \geq q_{next} + \epsilon / \bar{F}(q_{next})$.
Then we can increase $p(t_i)$ by $\epsilon$, and decrease $q(t_i)$ by $\epsilon / \bar{F}(q_{next})$. First, prices remain non-negative.
Second, the expected utility of type $t_i$ remains the same: the first stage utility decreases by $\epsilon$, and the second stage expected utility increases by $\epsilon / \bar{F}(q_{next}) \cdot Pr[ v^{(2)} \geq q(t_i) - \epsilon/\bar{F}(q_{next}) ] = \epsilon$ (where we used the fact that $\bar{F}(q_{next}) = Pr[ v^{(2)} \geq q_{next} ] = Pr[ v^{(2)} \geq x ]$, for all $x \in [q_{next}, q(t_i)]$, since $q(t_i)$ is not on the support of the second stage).
Thus, truthfulness is preserved by Observation~\ref{obs:1}. 
Finally, the expected payment of type $t_i$ does not decrease, so expected revenue does not decrease.
\end{proof}

Essentially the same argument also proves monotonicity for the first-stage prices.

\begin{claim} [First-stage price monotonicity] \label{cla:ind_price-monotonicity}
There exists an optimum mechanism that satisfies Claims \ref{cla:ind_alloc-monotonicity}-\ref{cla:ind_first-day_supp}, and such that if $t_i$ and $t_j$ are both winning, then $v_i > v_j \implies p(t_i) \geq p(t_j)$.
\end{claim}
\begin{proof}
Similar to Claim \ref{cla:ind_first-day_supp}, if $p(t_i) < p(t_j)$ we can increase $p(t_i)$ and decrease $q(t_i)$.
The latter is nonzero by \eqref{eq:ind-IC-win_win}.
\end{proof}

Observe that Claim~\ref{cla:ind_price-monotonicity} implies that if $q(t_i) = 0$ for some winning type, then for all (winning types) $t_j > t_i$, $p(t_j) = p(t_i)$ (since we cannot offer a lower price for the second stage). Building on the price monotonicity, we can therefore use another brute-force/enumeration step to further reduce our problem to the case where we know which winning types have $p(t_i)=v_i$ and which have $q(t_i)=0$.
Thus, we can assume we know one of the prices for every winning type; we just need to find the other price for one of them.
In the following claim we show how to solve the problem exactly using the fact that some of the second-stage prices actually lie on the support of the distribution. 

\begin{claim} \label{cla:ind_second-day_supp}
In {\em every} optimum mechanism, at least two of the following three conditions are satisfied:
\begin{itemize}
	\item there exists a winning type $t_i$ such that $q(t_i)$ is on the support of the second-stage distribution 
			(and it is nonzero);
	\item the second-stage price for all the losing types (observe that it is always the same for all of them), $q(0)$,
			is on the support of the second-stage distribution;
	\item one of the constraints between a loser and a winner \eqref{eq:ind-IC-win_lose} is tight.
\end{itemize}
\end{claim}

\begin{proof}
Our proof uses another gradual price increase argument. As long as neither of the first two conditions is satisfied, we can gradually increase the second-stage prices for all types simultaneously. Doing this with the right proportions maintains the IC constraints. Furthermore the revenue {\em strictly increases}: the prices increase, but as long as we don't cross any price in the support, the probabilities of selling the item to each type remain the same. 

Once the losing (resp. one of the winning) type's price hits the support, 
we can continue to increase the price for the winning (resp. losing) types as long as the IC constraints between losing and winning types are loose. Then, either the second stage price will hit the support, or the IC constraint will become tight, satisfying the second of the three conditions.
\end{proof}

Now, given any two of the three conditions in Claim \ref{cla:ind_second-day_supp},
we can find the all the optimum prices in polynomial time.
If the third condition is false, then we can enumerate to find a winning type and a losing type with prices in the support (that is, guess a winning type, and losing type, and two supported values as prices); then we can compute the induced prices for all other winning types using the IC constraints \eqref{eq:ind-IC-win_win}. 
If either of the two first conditions is false, then we can find optimum prices for one type (winning or losing), and then compute the remaining prices  through all the tight IC constraints. 
This completes the proof of Theorem \ref{thm:independent}.
\qed